\documentclass[a4paper,11pt]{article}
\usepackage{fullpage}
\usepackage{times}
\usepackage{soul}
\usepackage{url}
\usepackage[utf8]{inputenc}
\usepackage[small]{caption}
\usepackage{graphicx}
\usepackage{xcolor}
\usepackage{amsmath}
\usepackage{booktabs}
\usepackage{algorithm}
\usepackage{dsfont}
\usepackage{tabularx}
\usepackage{amsthm}
\usepackage{amssymb}
\usepackage{algpseudocode}

\usepackage{bbm}
\usepackage{graphicx}
\usepackage{color}
\usepackage{xcolor}
\usepackage{tcolorbox}
\usepackage{tikz}
\usetikzlibrary{positioning}
\usetikzlibrary{shapes.geometric, arrows}

\newtheorem{definition}{Definition}
\newtheorem{theorem}{Theorem}[section]
\newtheorem{corollary}[theorem]{Corollary}

\newtheorem{lemma}[theorem]{Lemma}

\usepackage{authblk}
\usepackage{enumitem}

\title{\bf An Improved Algorithm for Sparse Instances of SAT
}

\author[1]{Sanjay Jain}
\author[2]{Tzeh Yuan Neoh}
\author[1,3]{Frank Stephan }

\affil[1]{School of Computing, National University of Singapore, Singapore 117417}
\affil[2]{Agency for Science, Technology and Research, Singapore 138632}
\affil[3]{Department of Mathematics, National University of Singapore, Singapore 119076}

\date{\vspace{-10mm}}
\begin{document}
\maketitle
\begin{abstract}
    We show that the CNF satisfiability problem (SAT) can be solved in time $O^*(1.1199^{(d-2)n})$, where $d$ is either the maximum number of occurrences of any variable or the average number of occurrences of all variables if no variable occurs only once. This improves upon the known upper bound of $O^*(1.1279^{(d-2)n})$ by Wahlstr$\ddot{\text{o}}$m (SAT 2005) and $O^*(1.1238^{(d-2)n})$  by Peng and Xiao (IJCAI 2023). For $d\leq 4$, our algorithm is better than previous results. Our main technical result is an algorithm that runs in $O^*(1.1199^n)$ for 3-occur-SAT, a restricted instance of SAT where all variables have at most 3 occurrences. Through deeper case analysis and a reduction rule that allows us to resolve many variables under a relatively broad criteria, we are able to circumvent the bottlenecks in previous algorithms.
\end{abstract}
\section{Introduction}
The Boolean satisfiability problem (SAT) is the problem of deciding the satisfiability of formulas in conjunctive normal form (CNF). As the first discovered NP-complete problem \cite{cook1971complexity}, SAT and its many variants and extensions remain some of the most extensively studied NP-complete problems. Beyond its theoretical significance, the advent of SAT-solvers \cite{heule2011cube} has led to practical applications in computer assisted proofs \cite{cas}, AI planning and especially software verification \cite{marques2008practical}. 

While SAT-solvers work well in practice, their worst-case upper bound is the trivial bound of $O^*(2^n)$. Despite decades of research, no algorithm for SAT with worst-case runtime of $O^*(c^n)$ for $c < 2$ is known. In fact, the widely believed Strong Exponential Time Hypothesis \cite{impagliazzo2001complexity} conjectures that such an algorithm do not exist. Consequently, a substantial amount of work has been done on faster algorithms for restricted instances of SAT. For example, $k-$SAT is a restricted instance of SAT where every clause has at most $k$ literals. When $k = 3$, this is the famous 3-SAT problem. 3-SAT can be solved by a deterministic algorithm in time $O^*(1.32793^n)$ \cite{liu2018chain} and by a randomised algorithm in time $O^*(1.30698^n)$ \cite{scheder2022ppsz}.

The main focus of this paper are sparse instances of SAT. This include  instances where the maximum number of occurrences of any variable in a formula is at most $d$. The problem can be solved in linear-time when $d = 2$ and becomes NP-complete when $d\geq3$ \cite{tovey1984simplified}. 3-occur-SAT is the restricted instance of SAT when $d = 3$. The first non-trivial bound for 3-occur-SAT was achieved by Kullmann and Luckhardt \cite{kullmann1997deciding} with an algorithm that runs in time  $O^*(1.1299^n)$.  Subsequently, Wahlstr$\ddot{\text{o}}$m \cite{wahlstrom2005faster} and later Peng and Xiao \cite{inproceedings} presented algorithms for sparse formulas with upper bounds of $O^*(1.1279^{(d-2)n})$ and $O^*(1.1238^{(d-2)n})$ respectively. By setting $d = 3$, we can see that an $O^*(\alpha^{(d-2)n})$ algorithm for SAT is an $O^*(\alpha^n)$ algorithm for 3-occur-SAT.

An alternative approach  for tackling sparse instances of SAT are algorithms measured against the formula length $L$. This problem has been extensively studied \cite{van1988satisfiability, kullmann1997deciding, hirsch1998two, hirsch2000new} and more recent improvements of the upper bound include an  $O^*(1.0663^L$) algorithm by \cite{wahlstrom2005algorithm}, an $O^*(1.0652^L)$ algorithm by \cite{chen2009improved} and  an $O^*(1.0638^L)$ algorithm by  \cite{article} .  Notably, these latest attempts have all incorporated Wahlstr$\ddot{\text{o}}$m's algorithm \cite{wahlstrom2005faster}  as a 3-occur-SAT sub-procedure, with the runtime of the 3-occur-SAT procedure often being one of the bottleneck cases.

For the maximum satisfiability problem (MAXSAT), the restricted instance when $d = 3$ has also been extensively studied. MAXSAT is the optimisation variant of SAT where the objective is to satisfy the maximum number of clauses. More recent improvements of the upper bound include an  $O^*(1.194^n$) algorithm by \cite{xu2019resolution}, an $O^*(1.191^n)$ algorithm by \cite{belova2020algorithms} and  an $O^*(1.1749^n)$ algorithm by  \cite{brilliantov2023improved}.

\textbf{Our Contributions.} In this paper, we improve the  worst-case upper bound for 3-occur-SAT to $O^*(1.1199^n)$. Our algorithm (refer to section 3), like Wahlstr$\ddot{\text{o}}$m \cite{wahlstrom2005faster} and Peng and Xiao \cite{inproceedings}, is a modified branch-and-reduce algorithm. We apply reduction and branching rules until our formula acquires sufficient structure to enable the use of a fast 3-SAT algorithm by Beigel and Eppstein \cite{492575} as a sub-procedure.  For branch-and-reduce algorithms targeting variants and restricted instances of SAT, the key to improving the upper bound often hinges on effectively managing the cases with poor branching factors. For 3-occur-SAT, these critical bottlenecks for the previous attempts (directly or indirectly) are the branching rules to make the formula monotone and to handle all-negative clauses of length 2 and 3. 

For handling all-negative clauses of length 3 and greater, a more comprehensive analysis, by strategically selecting which variable to branch on and delving into additional cases, proves sufficient. The other bottlenecks are more challenging and requires new algorithmic ideas. We introduced a new reduction rule (step 6d) that generalizes many common reduction rules and allows us to resolve many variables at once. This new reduction rule enables us to impose significantly more structure of our formula. Hence, after branching on one variable, a reduction rule based on an autarkic set \cite{chu2021improved}  (step 8) allow us to achieves monotonicity. To effectively handle all-negative clauses of length two, we also resort to branching on two variables (step 9) when branching on one does not achieve a desirable branching factor.

Finally, we extend our algorithm to an $O^*(1.1199^{(d-2)n})$ algorithm for SAT. 
\begin{table}\centering
\begin{tabular}{c c  c } 
 \hline
 $d= 3$& $d = 4$ & Reference  \\ [0.5ex]
 \hline
 $O^*(1.1299^n)$ & $O^*(1.2766^n)$ & Kullmann and Luckhardt \cite{kullmann1997deciding} \\ 
 $O^*(1.1279^n)$ & $O^*(1.2721^n)$ &Wahlstr$\ddot{\text{o}}$m  \cite{wahlstrom2005faster} \\  
 $O^*(1.1238^n)$ & $O^*(1.2628^n)$ & Peng and Xiao \cite{inproceedings} \\ 
 $O^*(1.1199^n)$ & $O^*(1.2541^n)$ &\textbf{This paper} \\ 
\hline
\end{tabular}
  \caption{Progress for SAT with $d = 3,4$}
\end{table}

\section{Preliminaries}
A Boolean variable $x$ has 2 corresponding literals: the positive literal $x$ and the negative literal $\neg x$. A clause is a disjunction ($\lor$ operator) of literals and a CNF formula is a conjunction ($\land$ operator) of clauses. Let $F$ be a Boolean CNF formula. $Vars(F)$ is the set of variables in $F$. If the literal $x$ occurs in $i$ clauses and literal $\neg x$ occurs in $j$ clauses, the variable $x$ is a ($i,j$) variable. If $j = 0$, the literal $x$ is considered a pure literal. Without loss of generality, assume that $i \geq j$, flipping the signs of the literals when necessary. The degree, $deg(x)$, of the variable $x$ is the number of occurrences of either the literal $x$ or $\neg x$. For $F$ to be an instance of 3-occur-SAT, $deg(x)\leq 3$ for all $x\in Vars(F)$. A formula is $k$-regular if  $deg(x)=k$ for all $x\in Vars(F)$.

For a clause $C$, $|C|$ is the length of the clause which is the number of literals contained in $C$. A $k$-clause is a clause of length $k$. For $V \subseteq vars(F)$, $F_{contain \, V}$ is the set of clauses containing a variable in $V$ and  $F_{exclude \, V}$ = $F / F_{contain \, V}$. An \emph{all-positive clause} is a clause that contains only positive literals and a \emph{all-negative clause} is a clause that contains only negative literals. A formula $F$ is monotone if every clause is either all-positive or all-negative. 

For a clause \(C\), a subclause of $C$ is a subset of the literals in $C$. For a clause $C$ and set of variables $V \subseteq Vars(F)$, the maximal subclause of $C$ without $V$ is the largest subclause in $C$ without a variable in $V$, $maximal(C, V) = \{ l \mid l \in C \ \text{s.t.} \ l \notin V \ \text{and} \ \neg l \notin V\}$.

Lowercase letters are used to indicate variables and literals, while uppercase letters and Greek letters $\{\alpha,\beta,\gamma,\delta,\omega,\sigma\}$ are used to indicate (possibly empty) subclauses.

The neighbour of a \emph{literal} $l$,  $N(l)$, is the set of \emph{variables} that share a clause with the $\textit{literal}$ $l$. If $l$ is a literal of a variable $x$, $F[l]$ is the formula after assigning the variable $x$ to 1 if $l$ is positive or to 0 if $l$ is negative. This process removes every clause containing the literal $l$ and every instance of the literal $\neg l$. $F$[$l_1, \dots, l_k$] is the process repeated for the literals $l_1, \dots, l_k$.

Resolution \cite{davis1960computing} is a classic technique for SAT. For clauses $(x \lor c_1$), $(\neg x \lor c_2)$, the resolvent of the 2 clauses by $x$ is the clause $(c_1\lor c_2)$. For a formula $F$ and a variable $x$, the formula after resolving via $x$ adds all nontrivial resolvents by $x$ to $F$ and removes all clauses containing the variable $x$.

For our algorithm,  a formula is step-$k$ reduced if steps 1 to $k$ of the algorithm do not apply to $F$.
\subsection{Method of Analysis}
We use Kullmann's method from \cite{KULLMANN19991} to analyse the running time of our algorithm. We use the number of variables in $F$, $|Vars(F)|$, as our measure of complexity. Suppose a particular branching has $k$ branches that eliminates $n_1, \dots, n_k$ variables in those branches. The branching factor of this, $\tau(n_1, \dots, n_k)$, is the unique positive root of the equation $\sum_{i} x^{-n_i} =1$. If $\alpha$ is the largest branching factor of all branching conducted by the algorithm, then the algorithm runs in time $O^*(\alpha^n)$.  

Our largest branching factor is $\tau(3,11)$. The largest branching factors of Wahlstr$\ddot{\text{o}}$m \cite{wahlstrom2005faster} and Peng and Xiao\cite{inproceedings} algorithms are $\tau(4,8)$ and $\tau(5,7)$ respectively.

\begin{table}\centering
\begin{tabular}{c  c } 
 \hline
 Branching factors & Approximate Value  \\ [0.5ex] 
 \hline
 $\tau(6,7)$ & 1.11278 \\ 

 $\tau(5,8)$ & 1.1148  \\

 $\tau(4,9)$ & 1.11925  \\
  $\tau(3,11)$ & 1.11984  \\

 $\tau(6,11,14)$ & 1.11984  \\ [1ex] 
\hline
\end{tabular}
 \caption{Relevant branching factors of the algorithm}
\end{table}

\section{Algorithm}

Our algorithm employs a recursive approach, systematically applying the first applicable reduction or branching rule. After each application, the process restarts from the beginning to check for any further applicable steps. Steps 1 and 2 represents the base case for the recursive algorithm. Steps 3-6 are the reduction rules.  Step 7 is our primary branching step. Step 8 is a reduction rule that makes $F$ monotone. Step 9 is our secondary branching step. Step 10 uses the fast 3-SAT algorithm by \cite{492575}.

Our algorithm maintains the invariant that, at the conclusion of any step, $F$ has at most 3 occurrences per variable. Although variable occurrences may exceed three during the resolution, we tackle those specific cases to guarantee that this invariant is restored at the conclusion of each step. 

\medskip\noindent\newline\textbf{Algorithm 3-occur-SAT(F)}
\begin{enumerate}
  \item If $F$ is empty, return True. If $F$ contains an empty clause, return False.
  \item If every clause in $F$ contains a positive literal, return True, as all clauses can be satisfied if we set all variables to 1.  
  \item If $F$ is not standardised, standardise $F$ (see Definition 1).
  \item  The following reduction rules are sufficient to eliminate all variables that have less than 3 occurrences:
  \begin{enumerate}
      \item  Assign all ($k$,0) variables to 1. 
      \item Resolve all (1,1) variables.
  \end{enumerate}
  \item If there are 2 variables, $x$ and $y$, that occur together in 2 clauses, apply one of the following cases:
  \begin{enumerate}
    \item With clauses ($x\lor y\lor \alpha)$, ($x\lor y\lor \beta)$, ($\neg x\lor \gamma$) and ($\neg\,y\lor \delta$), resolutions via $x,y$ gives ($\alpha\lor \gamma\lor \delta$) and ($\beta\lor \gamma\lor \delta$). Introduce a new variable $z$ to get ($z\lor \alpha)$,  $(z\lor \beta)$ and ($\neg z\lor \gamma\lor \delta$).
    \item With clauses $(x\lor  y\lor \alpha)$, $(x\lor \neg\,y\lor \beta)$, ($\neg x\lor \gamma$) and $(y\lor \delta$), resolution via $x,y$ gives ($\alpha\lor \beta\lor \gamma$) and ($\delta\lor \beta\lor \gamma$). Introduce a new variable $z$ to get ($\alpha\lor $ $z$), ($\delta\lor $ $z$) and ($\neg z\lor \beta\lor \gamma$).
    \item With clauses  $(x\lor  y\lor \alpha)$, ($\neg x\lor \neg\,y\lor \beta)$, $(x\lor \gamma$) and ($y\lor \delta$), resolution via $x,y$ gives  ($\gamma\lor \beta\lor \delta$).
    \item With clauses  ($\neg x\lor  y\lor \alpha)$, $(x\lor \neg\,y\lor \beta)$, $(x\lor \gamma$) and $(y\lor \delta)$. Setting $x$ = 1 and $y$ = 1 satisfy all the clauses above.
  \end{enumerate}
  \item The following reductions rules help to tackle all-negative clauses:
    \begin{enumerate}
    \item With clauses  ($\neg x\lor \neg\,y\lor \alpha)$, $(x\lor  u\lor \beta$), $(y\lor \neg u\lor \gamma$), $(u\lor  \,\delta$), resolution via $u$ gives us  ($\neg x\lor \neg\,y\lor \alpha)$, $(x\lor  y\lor \beta\lor \gamma$), $(y\lor \gamma\lor \delta$). Notice that the clause $(x\lor  y\lor \beta\lor \gamma$) is redundant as if $F$ is satisfiable, there will be a satisfying assignment where one of $x$ or $y$ is assigned to 1.
    \item With clauses ($\neg\,x\lor \neg\,y\lor \gamma)$, ($x\lor u$), ($y\lor u$), ($x\lor \alpha$) and ($y\lor \beta$), after resolving via $x$ then $y$, we are left with the clauses ($u\lor \beta\lor \gamma)$, $(u\lor \alpha\lor \gamma)$, $(u\lor \gamma)$, $(\alpha\lor \beta\lor \gamma)$. The first 2 of these clauses are redundant, as they subsume $(u\lor \gamma)$. Introduce a new variable $z$ to get $(u\lor z)$, $(\alpha\lor \beta\lor z)$ and $(\neg z\lor \gamma)$.
    \item With clauses ($\neg\,x\lor \neg\,y)$, ($x\lor u$), ($y\lor u\lor \beta$), ($x\lor \alpha$) and ($y\lor \gamma$), after resolving via $x$ and $y$, we are left with clauses $(u\lor \beta)$,$(\alpha\lor \gamma)$, $(u\lor \gamma)$ and $(u\lor \alpha\lor \beta)$ (redundant). Introduce a new variable $z$ to get $(u\lor \beta)$, $(u\lor z)$, $(\alpha\lor z)$ and $(\neg z\lor \gamma)$.
    \item For all possible sets of variables $V$,  $|V|\leq 10$, perform the safe resolution  of variables in $V$  if the safe resolutions of variables in $V$ decreases the number of variables in $F$.  \end{enumerate}
  \item If the above steps do not apply, there exists an all-negative clause, ($\neg x_1 \lor \dots \lor \neg x_k$). If there exists $w\in \{x_1,\dots,x_k\}$ such that  the branching factor of 3-occur-SAT($F$[$\neg w$]) $\lor$ 3-occur-SAT($F$[$w$]) is at most 1.1199,  return 3-occur-SAT($F$[$\neg w$]) $\lor$ 3-occur-SAT($F$[$w$])
  \item Let $S$ be the set of all variables not in an all-negative clause. Assign all variables in $S$ to 1.
  \item For all-negative clause ($\neg x \lor \neg y$). If there exists $z \in N(x)$ such that the branching factor of 3-occur-SAT($F$[$\neg x, z$]) $\lor$ 3-occur-SAT($F$[$\neg x,\neg z$]) $\lor$ 3-occur-SAT($F$[$x$])  is at most $1.1199$, return 3-occur-SAT($F$[$\neg x, z$]) $\lor$ 3-occur-SAT($F$[$\neg x,\neg z$]) $\lor$ 3-occur-SAT($F$[$x$]). 
  \item Use the fast 3-SAT algorithm by \cite{492575}. 
\end{enumerate}
\textbf{End Algorithm}
\begin{definition}
Standardization of a CNF formula refers to applying the following reductions as far as possible
\end{definition}
\begin{enumerate}
    \item  Subsumption: If there are clauses $C_1$, $C_2$ such that $C_1$ is a subclause of $C_2$, remove $C_2$ from $F$
    \item  Trivial Clauses: Remove all clauses that contain a positive and negative literal of the same variable. i.e. ($x \lor \neg x \lor \alpha)$
    \item  Multi-occurring literal, Remove repeat instances of literals in a clause. i.e. ($x \lor x \lor \alpha) \rightarrow (x \lor \alpha)$ 
    \item Unit Clauses: For all clauses $C$ such that $|C| = 1$, let $C = (l)$. Reduce $F$ to $F[l]$
\end{enumerate}

\begin{definition}[Safe resolution]
For a subset $V$, a safe resolution with respect to $V$ involves resolving all variables in $V$ and introducing new variables to ensure no variable has more than 3 occurrences. 
\end{definition}
After resolution, the number of occurrences of some variables may increase. We thus introduce new variables to form an equisatisifable formula while decreasing the number of occurrences of all variables back to 3.  This can be done via the following process.
\begin{enumerate}
\item Standardize the formula 
    \item  Find the largest subclause, $C$, that appears in 2 separate clauses. If $|C|\geq 2$, for clauses ($C\lor\alpha$), ($C\lor\beta$), introduce a new variable $i$ to create clauses ($C\lor \neg i$), ($\alpha\lor i$) and ($\beta\lor i$)
    \item  Find the variable, $x$, with the largest degree. If $deg(x)\geq 4$, for clauses ($x\lor\alpha$), ($x\lor\beta$), introduce a new variable $i$ to create clauses ($x\lor \neg i$), ($\alpha\lor i$) and ($\beta\lor i$)
\end{enumerate}

Each time we introduce a new variable, we decrease the number of occurrences of a subclause.
\subsection{Paper Overview}

 In Section 4, we state the structure imposed on our formula by the reduction rules (step 3-6). In Section 5, we list sufficient conditions for us to apply our primary branching rule (step 7). In Section 6, we show that enough structure has been imposed that we can make the formula monotone simply by removing all clauses that contain a variable not in an all-negative clause (step 8).  In Section 7, we continue to tackle all-negative clauses of length 2 with the aid of monotonicity and a secondary branching rule (step 9). In Section 8, we show that if all the previous cases do not apply, the fast 3-SAT algorithm by \cite{492575} is applicable and runs in time $O^*(1.1092^n).$

\section{Step 1 - 6: Reduction Rules}
 We first note that if any of the reduction rules apply, we decrease the number of variables in $F$ by at least 1. This prevents endless loops between reduction rules. In this section, we describe the structure imposed by the reduction and give a sufficient condition for step 6d to be applicable.

\begin{lemma}\label{3}   If $F$ is a step-6 reduced formula, then all variables are (2,1) variables. \end{lemma}
\begin{proof}
    If there is a (3,0) variable  or a variable with 2 or less occurrences, it will be eliminated by step 4 of the algorithm. By convention, the literal $x$ has at least as many occurrences as the literal $\neg x$. Hence, all variables are (2,1) variables.
\end{proof}

\begin{lemma}\label{many rules}  Let $F$ be a step-6 reduced formula, with an all-negative clause $(\neg x \lor \neg y \lor \gamma)$ along with ($x \lor \alpha$) and ($y \lor \beta$). Then, if  $z \in Vars(\alpha)$ and $z\in Vars(\beta)$, (i) the literal $\neg z$ is in neither $\alpha$ or $\beta$ (ii) $|\alpha|\geq2$ or $|\beta|\geq2$ (iii) If $|\gamma|=0$, $|\alpha|\geq2$ and $|\beta|\geq2$\end{lemma}
\begin{proof}
    Step 6a, 6b and 6c of the algorithm ensures condition (i), (ii) and (iii) respectively. \end{proof}

 \begin{lemma}\label{at most 2} 
    Let $F$ be a step-6 reduced formula, with an all-negative clause $(\neg x_1 \lor \dots \lor \neg x_k)$. For all variables $v$, clauses containing $x_1, \dots, x_k$ contain at most 2 instance of a variable $v$.
 \end{lemma}
 \begin{proof}
     By lemma \ref{many rules}, if the negative literal $\neg v$ shares a clause with any of $x_1, \dots, x_k$, then the literal $\neg v$ does not share a clause with any of $x_1, \dots, x_k$. Thus, clauses containing $x_1, \dots, x_k$ contain at most 2 instance of a variable $v$. 
 \end{proof}

\begin{lemma}\label{key reduction}
    If there exists a subset $V$, $|V|>3$ and $|\{maximal(C,V) \mid C \in F_{contain V})\} \setminus \emptyset| \leq 3$, then after the safe resolution of variables in $V$, the number of variables in F decreases.
\end{lemma} \begin{proof}
    Suppose the 3 maximal subclauses are $c_1, c_2$ and $c_3$. After resolving all the variables in $V$, the clauses remaining are $F_{exclude \, V}$ and clauses that only contain any combination of $c_1, c_2$ and $c_3$. Let this set of clauses be $C$, $C \subseteq\mathcal{P}(\{ c_1,c_2,c_3\})$. If there are clauses $D_1,D_2 \in C$ such that $D_1 \subsetneq D_2$, we can remove $D_2$ from $C$. Thus, $C$ is an antichain of $\mathcal{P}(\{ c_1,c_2,c_3\})$. For $c_i \in \{c_1,c_2,c_3\}$, we note that all antichains of $\mathcal{P}(\{ c_1,c_2,c_3\})$ contains at most 2 subsets that contains $c_i$. Hence, we only need to introduce 1 additional variable for $c_i$ to have only 1 occurrence. 
    
    Thus, as we need to introduce at most 3 additional variables such that each of $c_1,c_2,c_3$ have only 1 occurrence, if $|V| > 3$, the safe resolution of variables in $V$ decrease the number of variables.

    \end{proof}
\begin{figure}[h]
    \centering
    \begin{tikzpicture}
        % Define the boxes
        \node[draw, rectangle, minimum width=3.5cm, minimum height=6cm] (box1) at (0,0) {
            \begin{tabular}{c}
                \textbf{Stage 1: }\\ $(\neg x \lor \neg y \lor C_1)$ \\ $(x \lor a)$ \\ $(x \lor b)$ \\ $(y \lor c)$ \\ $(y \lor d)$ \\ $(a \lor b )$ \\ $(c \lor d)$ \\ $(\neg a \lor \neg c)$ \\ $(\neg b \lor C_2)$\\ $(\neg d \lor C_3)$\end{tabular}
        };

        \node[draw, rectangle, minimum width=3cm, minimum height=4cm, right=1.5cm of box1] (box2) {
            \begin{tabular}{c}
               \textbf{Stage 2: }\\ $(C_1 \lor C_2 \lor C_3)$\\ $(C_1 \lor C_2)$\\ $(C_2 \lor C_3)$\\ $(C_1 \lor C_3)$\end{tabular}
        };

        \node[draw, rectangle, minimum width=3cm, minimum height=3cm, right=1.5cm of box2] (box3) {
            \begin{tabular}{c}
               \textbf{Stage 3: }\\ $(C_1 \lor C_2)$\\ $(C_2 \lor C_3)$\\ $(C_1 \lor C_3)$\end{tabular}
        };

        \node[draw, rectangle, minimum width=3.5cm, minimum height=4cm, right=1.5cm of box3] (box4) {
            \begin{tabular}{c}
               \textbf{Stage 4: }\\ $(v_1 \lor v_2)$\\ $(v_2 \lor v_3)$\\ $(v_1 \lor v_3)$\\ $(\neg v_1 \lor C1)$\\ $(\neg v_2 \lor C_2)$\\ $(\neg v_3 \lor C_3)$\end{tabular}
        };

        % Add arrows between boxes
        \draw[->] (box1.east) -- (box2.west);
        \draw[->] (box2.east) -- (box3.west);
        \draw[->] (box3.east) -- (box4.west);
    \end{tikzpicture}
    \caption{Example of Step 6d of the algorithm and lemma \ref{key reduction}. In stage 1, we have $F_{contain \{x,y,a,b,c,d\}}$ and the maximal subclauses (without $ \{x,y,a,b,c,d\}$) $C_1, C_2, C_3$. In stage 2, we have the clauses after resolving variables in $\{x,y,a,b,c,d\}$. In stage 3, we remove the redundant clause after standardization. In stage 4, we introduce 3 new variables $v_1,v_2,v_3$ to reduce the occurrences of our maximal subclauses back to 1. }
\end{figure}
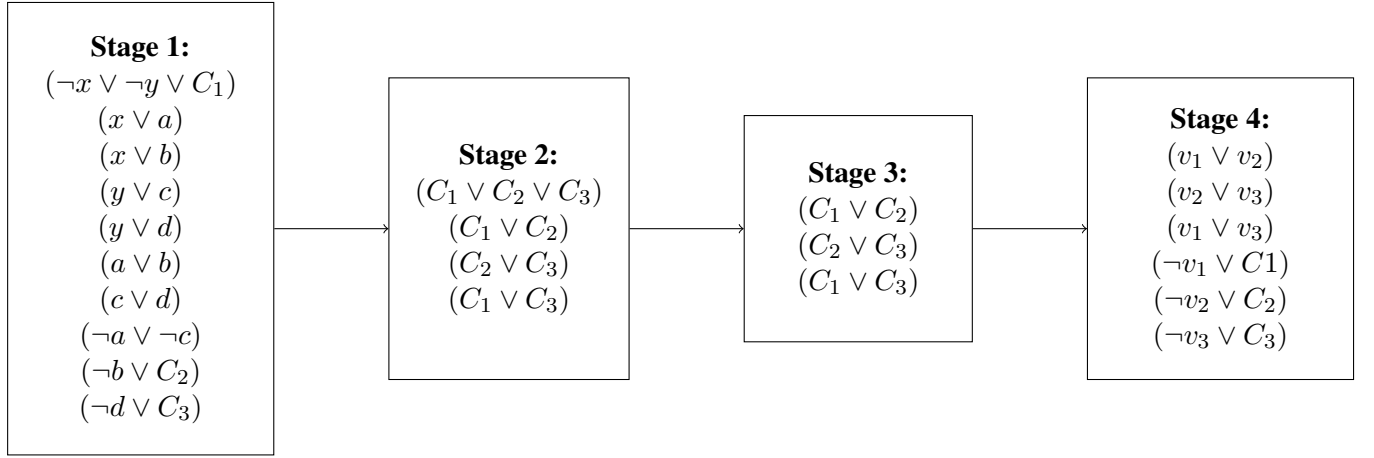

\begin{lemma}
    If there exists a subset $V$ such that $F_{contain V}$ contains variables only in $V \cup U$, if $|V| \geq 4$ and $|U| \leq 3$, then after the safe resolution of variables from some subset of $V \cup U$, the number of variables in F decreases.
\end{lemma}
\begin{proof}
    Let $U' \subseteq U$ be the set of variables in $U$ that have 2 or more occurrences in $F_{contain V}$. For each clause $C \in F_{contain V \cup U'}$, $maximal(C,V) \neq \emptyset$ if (1) $C \in F_{contain V}$ and there is a variable in $U\setminus U'$ that is in $C$, or (2) $C \notin F_{contain V}$ and there is a variable in $U'$ that is in $C$. At most $|U \setminus U'|$ clauses are applicable for the former case as all variables in $U \setminus U'$ have at most 1 occurrence in $F_{contain V}$. At most $| U'|$ clauses are applicable for the latter case as all variables in $ U'$ have at most 1 occurrence not in $F_{contain V}$. Thus, $|\{maximal(C,V \cup U') \mid C \in F_{contain V \cup U'})\} \setminus \emptyset| \leq 3 $ and by lemma \ref{key reduction}, after the safe resolution of variables in $V \cup U'$, the number of variables in F decreases.
\end{proof}

\section{Step 7: Dealing with all-negative clauses}
In this section, for all negative clause $C = (\neg x_1\lor \ldots\lor \neg x_k)$, we  show several conditions that are sufficient to ensure there exists a $w\in \{x_1,\dots,x_k\}$ such that returning 3-occur-SAT($F$[$\neg w$]) $\lor$ 3-occur-SAT($F$[$w$]) has branching factor of at most 1.1199. We first give a lower bound on the number of variables eliminated in when branching on a variable in an all-negative clause, followed by delving into the cases where $|C| \geq 4, |C|=3$ and $|C| =2$.

 \begin{lemma}\label{x}
    In the branch $F$[$x$], the variable $x$ and the variables in $N(x)$ are eliminated.
    \end{lemma} \begin{proof}
        In the branch $F$[$x$], all clauses containing the literal $x$ is removed. All variables in $N(x)$ will have 2 or less occurrences and can be eliminated by step 3 or step 4 of the algorithm. 
    \end{proof}
\begin{lemma}\label{not x}  
    With clauses $(\neg x_1\lor \ldots\lor \neg x_k)$, ($x_1\lor\alpha$) and ($x_1\lor\beta$), in the branch $F$[$\neg x_1$], variables $x_1,\dots,x_k$ and variables in $N(x_2) \cup N(x_3) \cup \dots \cup N(x_k)$ are eliminated. Furthermore, variables in $\alpha$ or $\beta$ are eliminated if $|\alpha| = 1$ or $|\beta| = 1$ respectively.\end{lemma}
    \begin{proof}
In the branch $F$[$\neg x_1$], $x_2,\dots,x_k$ are now pure literals and are assigned to 1. All clauses they are involved in are eliminated and hence all variables in those clauses are eliminated. In the branch $F$[$\neg x_1$], the clauses ($x_1\lor\alpha$) and ($x_1\lor\beta$) becomes just $\alpha$ and $\beta$. If either $\alpha$ or $\beta$ contain just 1 variable, they become a one-literal clause and we can set the literal to 1.

\end{proof}

\subsection{Tackling all-negative clauses of length 4 or greater}

 \begin{lemma}\label{limit} 
    Let $F$ be a step-6 reduced formula, with an all-negative clause $(\neg x_1 \lor \dots \lor \neg x_k)$. Let $n$ be the number of 2-clauses the literals $x_2, \dots, x_k$ are in. Then $|N(x_2) \cup \dots \cup N(x_k)| \geq \max(n, (|N(x_2)| + \dots + |N(x_k)|)/2)$.
 \end{lemma}
\begin{proof}
By lemma \ref{many rules}, if there are clauses $(x_i \lor l_i)$ and $(x_j \lor l_j)$, then the variables of $l_i$ and $l_j$ are different. Thus, $|N(x_2) \cup \dots \cup N(x_k)| \geq n$.  By lemma \ref{at most 2}$, |N(x_2) \cup \dots \cup N(x_k)| \geq (|N(x_2)| + \dots + |N(x_k)|)/2$
\end{proof}

\begin{lemma}\label{C=4} 
    If $|C| \geq 4$, there is always a $w\in \{x_1,\dots,x_k\}$, such that returning 3-occur-SAT($F$[$\neg w$]) $\lor$ 3-occur-SAT($F$[$w$]) has a  branching factor of at most $\tau(4,9).$ \end{lemma}
\begin{proof}

Let $C$  = $(\neg x_1\lor \ldots\lor \neg x_k)$. Without loss of generality, let $|N(x_1)| \geq |N(x_2)| \geq \dots \geq |N(x_k)|$. 
\medskip\noindent\newline\textbf{Case 1:} $|N(x_1)|\geq 4$

In the branch $F[x_1]$, by lemma \ref{x}, we can eliminate 5 variables. Let $n$ be the number of 2-clauses the literals $x_2, \dots, x_k$ are in. If $n\leq 3$, $(|N(x_2)| + \dots + |N(x_k)|) \geq 9$. Thus, by lemma \ref{limit}, $|N(x_2) \cup N(x_3) \cup \dots \cup N(x_k)| \geq 4$. Returning 3-occur-SAT($F$[$x_1$]) $\lor$ 3-occur-SAT($F$[$\neg x_1$]) has branching factor of at most $\tau(5,8)$ 
\medskip\noindent\newline\textbf{Case 2:}  $|N(x_1)| = 3$

In the branch $F[x_1]$, by lemma \ref{x}, we can eliminate 4 variables. If $|N(x_2) \cup \dots \cup N(x_k)| \geq 5$,  by lemma \ref{not x}, we can eliminate 9 variables in the branch $F[\neg x_1]$. Let $n$ be the number of 2-clauses the literals $x_2, \dots, x_k$ are in. If $n \leq 3$ or $n \geq 5$,  $|N(x_2) \cup N(x_3) \cup \dots \cup N(x_k)| \geq 5$. 

Otherwise, $n = 4$. If $|N(x_2) \cup N(x_3) \cup \dots \cup N(x_k)| \leq 4$, we note that this is only possible if $k = 4$ and thus, $|N(x_2)| + |N(x_3)| + |N(x_4)| = 8$. As $|N(x_1)|=3$, there must be a clause of form $(x_1 \lor v_1)$. As all variables in   $N(x_2) \cup N(x_3) \cup N(x_4)$ share 2 clauses with $x_2,x_3,x_4$, we know that $v_1 \notin N(x_2) \cup N(x_3) \cup N(x_4)$. Thus,  by lemma \ref{not x}, we can eliminate 9 variables in the branch $F[\neg x_1]$
Returning 3-occur-SAT($F$[$x_1$]) $\lor$ 3-occur-SAT($F$[$\neg x_1$]) has branching factor of at most $\tau(4,9)$.
\medskip\noindent\newline\textbf{Case 3:} $|N(x_1)| = 2$  

By Lemma \ref{many rules}, we know that $|N(x_i) \cap N(x_j)| = 0$ if $i\neq j$. Thus, by Lemma \ref{not x}, we eliminate at least 12 variables in the branch $F[\neg x_1]$. Returning 3-occur-SAT($F$[$x_1$]) $\lor$ 3-occur-SAT($F$[$\neg x_1$]) has branching factor of at most $\tau(3,12)$ \end{proof}
\subsection{Tackling all-negative 3-clauses}
    
     Let the all-negative 3-clause be ($\neg x\lor \neg y\lor \neg z)$ and let there be clauses $(x\lor \alpha)$, $(x\lor \beta)$, $(y\lor \gamma)$, $(y\lor \delta)$, $(z\lor \omega)$ and $(z\lor \sigma)$.  Without loss of generality, let $|\alpha|\geq|\beta|$, $|\gamma|\geq|\delta|$, $|\omega|\geq|\sigma|$  and $|\alpha|+|\beta|\geq|\gamma|+|\delta|\geq|\omega|+|\sigma|$.

 Let A = $Vars$($\alpha$) $\cup$ $Vars$($\beta$), B =$Vars$($\gamma$) $\cup$ $Vars$($\delta$) and C =$Vars$($\omega$) $\cup$ $Vars$($\sigma$). From the inclusion-exclusion principle, we have $|Vars(A\,\cup\,B\,\cup\,C)| = |Vars(A)|  + |Vars(B)| + |Vars(C)| - |Vars(A\,\cap\,B)|  - |Vars(B\,\cap\,C)| - |Vars(A\,\cap\,C)|$. Note that $|Vars(A\,\cap\,B\,\cap\,C)| = 0$ by Lemma \ref{many rules}. 
 
 Before we proceed, we first state a lemma that is used repeatedly in the following cases.
\begin{lemma}\label{C = 3,reduction} 
With the clause ($\neg x\lor \neg y\lor \neg z)$, if in the branch $F$[$x$], we are left with the clause ($\neg y\lor \neg z)$ and there are clauses $(y\lor u\lor \gamma)$, $(z\lor u)$, we can eliminate 1 additional variable other than $x$ and $N(x)$ by step 6c of the algorithm.
\end{lemma}

\begin{lemma}\label{C=3} 
    If $|\alpha|+|\beta| \geq 4$, there exists $w\in \{x,y,z\}$ such that branching factor of 3-occur-SAT($F$[$\neg w$]) $\lor$ 3-occur-SAT($F$[$w$]) is at most $\tau(3,11)$. \end{lemma} 
\begin{proof} 
We delve into 3 cases, (1)  $|\gamma|+|\delta|\geq4$, (1)  $|\gamma|+|\delta|=3$, (3)  $|\gamma|+|\delta|=2$, 
\medskip\noindent\newline\textbf{Case 1:} $|\gamma|+|\delta|\geq4$ 

Without loss of generality, let $|Vars(A\,\cap\,C)| \geq |Vars(B\,\cap\,C)|$. Thus $|Vars(B\,\cup\,C)|\geq 5$, by Lemma \ref{not x}, the branch $F$[$\neg x$, $y$, $z$] eliminates at least 8 variables. Thus, returning 3-occur-SAT($F$[$x$]) $\lor$ 3-occur-SAT($F$[$\neg$$x$, $y$, $z$]) have a branching factor of at most $\tau (5,8)$.
\medskip\noindent\newline\textbf{Case 2:} $|\gamma|+|\delta|=3$ 

If $|Vars(B\,\cup\,C)|\geq 5$, then the branch $F$[$\neg x$, $y$, $z$] eliminates 8 variables and returning 3-occur-SAT($F$[$x$]) $\lor$ 3-occur-SAT($F$[$\neg$$x$, $y$, $z$]) have a branching factor of at most $\tau (5,8)$. Else, by the inclusion-exclusion principle, $|Vars(B\,\cap\,C)| \geq 2$  or $|Vars(B\,\cap\,C)| \geq 1$,  $|\sigma|+|\omega|$ = 2 and $|Vars(B\,\cup\,C)|\geq 4$ (as if $|Vars(B\,\cap\,C)| = 0$ then $|Vars(B\,\cup\,C)|\geq 5)$.

In the former case, if $|Vars(B\,\cap\,C)| \geq 2$, $|Vars(A\,\cup\,C)| \geq 6$ (as $|Vars(A\,\cap\,B\,\cap\,C)| = 0$). Returning 3-occur-SAT($F$[$y$]) $\lor$ 3-occur-SAT($F$[$x$,$\neg$ $y$, $z$]) have a branching factor of at most $\tau (4,9)$.

In the latter case, in the branch $F$[$x$], the clause ($\neg x \lor \neg y \lor \neg z$) becomes ($\neg y \lor \neg z$). Furthermore, there  are clauses ($y \lor u \lor \gamma$), ($z \lor u)$ and by step 6b of the algorithm, another variable is eliminated. Returning 3-occur-SAT($F$[$x$]) $\lor$ 3-occur-SAT($F$[$\neg$$x$, $y$,$z$]) have a branching factor of at most $\tau (6,7)$.
\medskip\noindent\newline\textbf{Case 3 :} $|\gamma|+|\delta|=2$.
We tackle this case through several sub-cases. We note that in this case, in the branch $F[\neg z]$, we eliminate all variables in $A \cup B \cup C$. Furthermore, as $|B| = |C| = 2$, by lemma \ref{many rules}, $|B \cap C| = 0.$
\medskip\noindent\newline\textbf{Case 3a :}  $|\alpha| \geq 3$

 We note that by lemma \ref{many rules}, $|B \cup C| = 4$ and $|B \cup C \cup vars(\beta)| = 5$ if $|\beta| = 1$. Thus, if $|\alpha| \geq 3$ and $|\beta| \geq 2$, then returning 3-occur-SAT($F$[$x$]) $\lor$ 3-occur-SAT($F$[$\neg$$x$, $y$,$z$]) have a branching factor of at most $\tau (6,7)$. If $|\alpha| \geq 3$ and $|\beta| = 1$, then returning 3-occur-SAT($F$[$x$]) $\lor$ 3-occur-SAT($F$[$\neg$$x$, $y$,$z$]) have a branching factor of at most $\tau (5,8)$. 
\medskip\noindent\newline\textbf{Case 3b :} $|\alpha| = |\beta| =2$, $|A \cup B \cup C| =8$. 
 
 If $|A \cup B \cup C| =8$, then by lemma \ref{not x}, returning 3-occur-SAT($F$[$z$]) $\lor$ 3-occur-SAT($F$[$x$, $y$,$\neg$$z$]) have a branching factor of at most $\tau (3,11)$. 
\medskip\noindent\newline\textbf{Case 3c :} $|\alpha| = |\beta| =2$, $|A \cup B \cup C| <8$.

Without loss of generality, let $|Vars(A\,\cap\,B)| \geq |Vars(A\,\cap\,C)|$. Thus,  $|Vars(A\,\cap\,B)| \geq 1$ and there are clauses $(x \lor u \lor l), (y \lor l)$. In the branch $F[z]$, we can eliminate at least 4 variables ($z, N(Z)$ and an additional variable due to step 6b of the algorithm).
\medskip\noindent\newline\textbf{Case 3ci :} $|\alpha| = |\beta| =2$,  $8 > |A \cup B \cup C| \geq 6$.

 Returning 3-occur-SAT($F$[$z$]) $\lor$ 3-occur-SAT($F$[$x$, $y$,$\neg$$z$]) have a branching factor of at most $\tau (4,9)$. 
\medskip\noindent\newline\textbf{Case 3cii :} $|\alpha| = |\beta| =2$,  $|A \cup B \cup C| < 6$.

$|Vars(A\,\cap\,B)|= 2$ and $|Vars(A\,\cap\,C)| \geq 1$. There are 2 possibilities of how $|A \cap B| \geq 2$, (1) there are clauses $(x \lor a \lor l_1), (x \lor b \lor l_2), (y \lor a), (y \lor b)$, or (2) there are clauses $(x \lor a \lor b),  (y \lor a), (y \lor b)$.

For case (1), in the branch $F[z]$, there are clauses $(x \lor a \lor C_1), (x \lor b \lor C_2), (y \lor a), (y \lor b)$, $(\neg x \lor \neg y)$. After resolving via $x,y$, we have the clauses $(a \lor C_1), (b \lor C_2)$. Note that the number of occurrences of $a,b$ decreases. Thus, we eliminate 7  variables in this branch ($\{x,y,z,a,b\}$ and the 2 neighbours of $z$ which are different from $a,b$ as $|B\cap C| = 0$). 
In the branch $F[x,y,\neg z]$, 7 variables are eliminated in a similar way. Returning 3-occur-SAT($F$[$z$]) $\lor$ 3-occur-SAT($F$[$x$, $y$,$\neg$$z$]) have a branching factor of at most $\tau (7,7)$. 

For case (2), as  $|Vars(A\,\cap\,C)| \geq 1$, there are clauses $(x \lor a \lor b),  (y \lor a), (y \lor b)$ and $(x \lor c \lor l_1), (z \lor c), (z \lor l_2), (\neg a \lor C_1), (\neg b \lor C_2)$ where $l_1$ may be $l_2$. We note that due to the clause $(x \lor a \lor b)$, the literals $\neg a, \neg b$ are not in the same clause. 

In the branch $F[y]$, there are clauses $(x \lor C_1 \lor C_2), (x \lor c \lor l_1), (z \lor c), (z \lor l_2).$ If $|vars(C_1 \cup C_2 \cup \{c,l_1,l_2\})| \geq 4$, then, in the branch $F[x,y,\neg z]$, we can eliminate 9 variables ( $|\{a,b,x,y,z\} \cup vars(C_1 \cup C_2 \cup \{c,l_1,l_2\})| \geq 9$). Then returning 3-occur-SAT($F$[$z$]) $\lor$ 3-occur-SAT($F$[$x$, $y$,$\neg$$z$]) have a branching factor of at most $\tau (4,9)$. 

Otherwise, $|vars(C_1 \cup C_2 \cup \{c,l_1,l_2\})| \leq 3$. Let $vars(C_1 \cup C_2 \cup \{c, l_1,l_2\}) = N(x) \cup N(z) \subseteq \{c,d,e\}.$ Through a simple counting argument, at most 3 instances of the variables $c,d,e$ are not part of the clauses $(\neg x \lor \neg y \lor \neg z), (x \lor a \lor b), (x \lor c \lor l_1), (y \lor a), (y \lor b), (z \lor c), (z\lor l_2), (\neg a \lor C_1), (\neg b\lor C_2).$ Therefore, Step 6d of the algorithm is applicable to $\{x,y,z,a,b,c,d,e\}$.

\end{proof}

\begin{lemma}\label{C=3ii} 
    If $|\alpha|+|\beta| = 3$, there exists $w\in \{x,y,z\}$ such that branching factor of 3-occur-SAT($F$[$\neg w$]) $\lor$ 3-occur-SAT($F$[$w$]) is at most $\tau(3,11)$. \end{lemma}

\begin{proof}
    We delve into 3 cases,  (1) $|\gamma|+|\delta|=|\omega|+|\sigma|=3$, (2) $|\gamma|+|\delta|=3, |\omega|+|\sigma|=2$ and (3) $|\gamma|+|\delta|=|\omega|+|\sigma|=2$.
\medskip\noindent\newline\textbf{Case 1:} $|\alpha|+|\beta|=|\gamma|+|\delta|=|\omega|+|\sigma|=3$  

Without loss of generality, let $ |Vars(B\,\cap\,C)|  \geq|Vars(A\,\cap\,C)|  \geq  |Vars(A\,\cap\,B)|$
\medskip\noindent\newline\textbf{Case 1a :} $|Vars(\alpha\,\cup\,\beta\,\cup\,\gamma\,\cup\,\delta\,\cup\,\omega\,\cup\,\sigma)|\geq$ 7 

As $|Vars(A\,\cap\,B)|$ =0, $|N(x) \cup N(y)| \geq 6$. Hence, by Lemma \ref{not x}, returning 3-occur-SAT($F$[$z$]) $\lor$ 3-occur-SAT($F$[$x$, $y$,$\neg$$z$]) have a branching factor of at most $\tau (4,9)$.
\medskip\noindent\newline\textbf{Case 1b :} $|Vars(\alpha\,\cup\,\beta\,\cup\,\gamma\,\cup\,\delta\,\cup\,\omega\,\cup\,\sigma)|$ = 6 

If $|Vars(A\,\cap\,B)|$ = 0, we can do the same as the previous case. Else,  $ |Vars(B\,\cap\,C)|  = |Vars(A\,\cap\,C)|  =  |Vars(A\,\cap\,B)| = 1$. Looking at the clauses,  $(x\lor  a \lor  b),\, (x\lor  c)$, 2 of $a$, $b$ or $c$ are neighbours with $y$ or $z$. If both $a$ and $b$ are neighbours of $y$ or $z$, then the branch $F$[$\neg$ $x$, $y$, $z$]  will eliminate 9 variables. Returning 3-occur-SAT($F$[$x$]) $\lor$ 3-occur-SAT($F$[$\neg$$x$, $y$,$z$]) have a branching factor of at most $\tau (4,9)$. Else, $c$ is a neighbour of $y$ or $z$. Without loss of generality, let $c$ be a neighbour of $y$. By Lemma \ref{C = 3,reduction}, we can eliminate an additional variable in the branch $F$[$z$]. As $|Vars(A\,\cap\,B)|$ = 1, $|N(x) \cup N(y)| = 5$, returning 3-occur-SAT($F$[$z$]) $\lor$ 3-occur-SAT($F$[$x$, $y$,$\neg$$z$]) have a branching factor of at most $\tau (5,8)$.
\medskip\noindent\newline\textbf{Case 1c :} $|Vars(\alpha\,\cup\,\beta\,\cup\,\gamma\,\cup\,\delta\,\cup\,\omega\,\cup\,\sigma)|$ = 5 

$ |Vars(B\,\cap\,C)|  = 2$, $|Vars(A\,\cap\,C)|  =  |Vars(A\,\cap\,B)| = 1$. 
Looking at the clauses,  $(x\lor  a \lor  b),\, (x\lor  c)$, 2 of $a$, $b$ or $c$ are neighbours with $y$ or $z$. If both $a$ and $b$ are neighbours of $y$ or $z$, then the branch $F$[$\neg$ $x$, $y$, $z$]  will eliminate 8 variables. 
Furthermore, by Lemma \ref{C = 3,reduction}, the branch $F$[$x$] will eliminate an additional variable. Returning 3-occur-SAT($F$[$x$]) $\lor$ 3-occur-SAT($F$[$\neg$$x$, $y$,$z$]) have a branching factor of at most $\tau (4,9)$. Else, $c$ is a neighbour of $y$ or $z$. Without loss of generality, let $c$ be a neighbour of $y$. By Lemma \ref{C = 3,reduction}, we can eliminate an additional variable in the branch $F$[$z$]. As $|Vars(A\,\cap\,B)|$ = 1, $|N(x) \cup N(y)| = 5$, returning 3-occur-SAT($F$[$z$]) $\lor$ 3-occur-SAT($F$[$x$, $y$,$\neg$$z$]) have a branching factor of at most $\tau (5,8)$.
\medskip\noindent\newline\textbf{Case 2:} $|\alpha|+|\beta|=|\gamma|+|\delta|=3,\,|\omega|+|\sigma|=2$ 

Without loss of generality, let $|Vars(A\,\cap\,C)| \geq |Vars(B\,\cap\,C)|$. As  $|\beta|=|\delta|=|\omega|=|\sigma|=1$, from Lemma \ref{many rules}, they all are distinct variables. 
\medskip\noindent\newline\textbf{Case 2a :} $|Vars(\alpha\,\cup\,\beta\,\cup\,\gamma\,\cup\,\delta\,\cup\,\omega\,\cup\,\sigma)|\geq 6$  

Looking at clauses $(x\lor \beta)$, $(y\lor \gamma)$, if the variable in $\beta$ is in $\gamma$,  by Lemma \ref{C = 3,reduction}, we can eliminate another variable in the branch $F$[$z$]. Returning 3-occur-SAT($F$[$z$]) $\lor$ 3-occur-SAT($F$[$x$, $y$,$\neg$$z$]) have a branching factor of at most $\tau (4,9)$.  If one of the variable in $\omega$ or $\sigma$ is in $\gamma$, by Lemma \ref{C = 3,reduction}, we can eliminate another variable in the branch $F$[$x$]. Furthermore, $|Vars(\beta\,\cup\,\gamma\,\cup\,\delta\,\cup\,\omega\,\cup\,\sigma)| \geq 5$. Hence, returning 3-occur-SAT($F$[$x$]) $\lor$ 3-occur-SAT($F$[$\neg x$, $y$, $z$]) have a branching factor of at most $\tau (5, 8)$. 
Else, $|Vars(\beta\,\cup\,\gamma\,\cup\,\delta\,\cup\,\omega\,\cup\,\sigma)| = 6$ as there is no overlap amongst the clauses.  Hence, returning 3-occur-SAT($F$[$x$]) $\lor$ 3-occur-SAT($F$[$\neg x$, $y$, $z$]) have a branching factor of at most $\tau (4, 9)$.
\medskip\noindent\newline\textbf{Case 2b :} $|Vars(\alpha\,\cup\,\beta\,\cup\,\gamma\,\cup\,\delta\,\cup\,\omega\,\cup\,\sigma)|$ = 5 

If $|Vars(A\,\cap\,C)| = 2$, $|Vars(A\,\cap\,B)| = 1$. We have clauses $(x\lor a \lor b)$, $(x\lor c)$, $(y\lor c \lor d)$, $(y\lor e)$, $(z\lor a)$, $(z\lor b)$. If the literal $e$ is not in a clause solely with either $\neg a$ or $\neg b$, in the branch $F$[$x, \neg y, z$], $e$ is assigned to 1 and the variables that neighbours the literal $e$ are eliminated. Returning $F$[$y$] $\lor$ $F$[$x, \neg y, z$] has branching factor of at most $\tau (5, 8)$. Else, if $e$ is in a clause solely with $\neg a$ or $\neg b$, in the branch $F$[$x$], we will have clauses $(z \lor e), (y \lor e)$ and by Lemma \ref{C = 3,reduction}, another variable can be eliminated. Returning $F$[$x$] $\lor$ $F$[$\neg x, y, z$] has branching factor of at most $\tau (5, 8)$. 

If $|Vars(A\,\cap\,C)| = 1$ and $|Vars(A\,\cap\,B)| = 2$, there are 2 possibilities, (1) there are clauses $(x\lor a \lor b)$, $(x\lor c)$, $(y\lor c \lor d)$, $(y\lor b)$, $(z\lor a)$, $(z\lor e)$, or (2) there are clauses $(x\lor a \lor b)$, $(x\lor c)$, $(y\lor c \lor b)$, $(y\lor d)$, $(z\lor a)$, $(z\lor e)$. In the former case, in the branch $F$[$z$], there are clauses ($\neg x\lor \neg y)$, $(x\lor a \lor b)$, $(x\lor c)$, $(y\lor c \lor d)$, $(y\lor b)$, after resolving $x,y$ we have $(a\lor b \lor c\lor d)$, $(a\lor b)$,  $(c\lor d)$ $(c\lor b)$ which can be simplified to $(a\lor b)$,  $(c\lor d)$ $(c\lor b)$. Hence, we can eliminate 2 additional variables and returning 3-occur-SAT($F$[$z$]) $\lor$ 3-occur-SAT($F$[$x$, $y$, $\neg$$z$]) have a branching factor of at most $\tau (5,8)$. In the latter case, by Lemma \ref{not x} and step 6c of the algorithm, returning 3-occur-SAT($F$[$y$]) $\lor$ 3-occur-SAT($F$[$x$, $\neg$$y$, $z$]) have a branching factor of at most $\tau (5,8)$. 

If $|Vars(A\,\cap\,C)| = |Vars(A\,\cap\,B)| = |Vars(B\,\cap\,C)| = 1$. One of $\alpha$ or $\gamma$ must only contain variables that are neighbours with another of $x$, $y$ or $z$. Without loss of generality, suppose that $\alpha$ only contain variables that are neighbours of $y$ or $z$. By Lemma \ref{not x} and step 6c of the algorithm, returning 3-occur-SAT($F$[$x$]) $\lor$ 3-occur-SAT($F$[$\neg$$x$, $y$, $z$]) have a branching factor of at most $\tau (5,8)$. 

Else, $|(A\,\cap\,B)| = 3$, After the branch $F$[$z$], we have clauses ($\neg x\lor \neg y)$, $(x\lor a\lor b)$, $(x\lor c)$, $(y\lor  a\lor c)$, $(y\lor b)$. Resolving via $x$ and $y$ gives us  $(a\lor b)$, $(b\lor c)$, $(a\lor c)$, eliminating 2 variables in this branch. Hence, by Lemma \ref{not x}, returning 3-occur-SAT($F$[$z$]) $\lor$ 3-occur-SAT($F$[$x$, $y$,$\neg$$z$]) have a branching factor of at most $\tau (5,8)$. 
\medskip\noindent\newline\textbf{Case 2c :} $|Vars(\alpha\,\cup\,\beta\,\cup\,\gamma\,\cup\,\delta\,\cup\,\omega\,\cup\,\sigma)|$ = 4 

We have clauses $(x\lor a\lor b)$, $(x\lor c)$, $(y\lor c\lor d)$, $(y\lor b)$, $(z\lor a)$, $(x\lor d)$ which can be resolved to $(a \lor b)$, $(c \lor d)$.
\medskip\noindent\newline\textbf{Case 3:} $|\alpha|+|\beta|=3,\,\,|\gamma|+|\delta|=|\omega|+|\sigma|=2$  

$|Vars(B\,\cap\,C)|$ = 0. 
Without loss of generality, let $|Vars(A\,\cap\,B)| \geq |Vars(A\,\cap\,C)|$ 
\medskip\noindent\newline\textbf{Case 3a :} $|Vars(\alpha\,\cup\,\beta\,\cup\,\gamma\,\cup\,\delta\,\cup\,\omega\,\cup\,\sigma)|$ = 7 

We have clauses $(x\lor a\lor b)$, $(x\lor c)$,$(y\lor d)$,$(y\lor e)$, $(z\lor f)$,$(z\lor g)$. Suppose one of $a$, $b$ or $c$ is negative and call this literal $\neg h$. If $h$ is only neighbours of $a$, $b$ or $c$, we can simplify the formula by resolving via $h$. Else, the branch $F$[$x$] eliminates at least 5 variables and returning 3-occur-SAT($F$[$x$]) $\lor$ 3-occur-SAT($F$[$\neg x$, $y$,$z$]) have a branching factor of at most $\tau (5,8)$. Suppose one of $d$, $e$, $f$ or $g$ is negative. Without loss of generality, let $d$ be negative. The branch $F$[$y$] eliminates at least 4 variables and returning 3-occur-SAT($F$[$y$]) $\lor$ 3-occur-SAT($F$[$x$, $\neg y$,$z$]) have a branching factor of at most $\tau (4,10)$. Assume now that all of $a,b,c,d,e,f,g$ are positive literals.

If the literal $c$ is neighbours with variable other than those in $\{x,y,z,c,d,e,f,g\}$ returning 3-occur-SAT($F$[$x$]) $\lor$ 3-occur-SAT($F$[$\neg x$, $y$,$z$]) have a branching factor of at most $\tau (4,9)$.

Suppose the literal $c$ shares a clause with $\neg d$ (symmetric if it is $\neg e$, $\neg f$ or $\neg g$),  in the branch $F$[$z$], we are left with clauses $(\neg x\lor\neg y)$, $(x\lor a\lor b)$, $(x\lor c)$,$(y\lor d)$,$(y\lor e)$ and ($c\lor\neg d \lor\gamma$). If $\gamma$ contains a variable outside of $e$ (after the resolution for variables $f,g$ in the branch $F[z]$), then the branch $F$[$\neg x,y$] will now eliminate 6 variables. Hence, returning 3-occur-SAT($F$[$x$]) $\lor$ 3-occur-SAT($F$[$\neg x$, $y$,$z$]) have a branching factor of at most $\tau (4,9)$. Otherwise, $\gamma$ is empty or only contains the variable $e$. Suppose there is clause ($d\lor\delta$), after resolution via $d$, we have clauses ($y \lor c \lor \gamma$), ($\delta \lor c \lor \gamma)$. If $\gamma$ is nothing, we have clauses $(x\lor c)$,$(y\lor c)$ and hence we can just do $F$[$c$] using step 6b. If $\gamma$ is $e$, ($y \lor c \lor e$) is redundant due to clause ($y \lor e$) and hence we have eliminated $y$ and $d$. If $\gamma$ is $\neg e$, ($y \lor c \lor \neg e$) can be reduced to ($y \lor c$) due to clause ($y \lor e$) and we can branch $F$[$c$] using step 6b. Hence, returning 3-occur-SAT($F$[$z$]) $\lor$ 3-occur-SAT($F$[$x$, $y$, $\neg z$]) have a branching factor of at most $\tau (5,10)$.

Suppose the literal $c$ shares a clause with $d$ (symmetric if it is $e$, $f$ or $g$),  in the branch $F$[$z$], we are left with clauses $(\neg x\lor\neg y)$, $(x\lor a\lor b)$, $(x\lor c)$,$(y\lor d)$,$(y\lor e)$ and ($c\lor d \lor\gamma$). Suppose there is the clause ($\neg d \lor \delta$). If $\delta$ is just $e$ or $\neg e$, after resolution of $d$, we either have duplicate clauses $(y\lor e)$ or we have the clauses $(y\lor e)$, $(y\lor \neg e)$. In either case, we can eliminate another variable and returning 3-occur-SAT($F$[$z$]) $\lor$ 3-occur-SAT($F$[$x$, $y$, $\neg z$]) have a branching factor of at most $\tau (4,9)$. Else, in the branch $F$[$\neg x,y$], $\neg d$ is a pure literal and the variable in $\delta$ is eliminated. Returning 3-occur-SAT($F$[$x$]) $\lor$ 3-occur-SAT($F$[$\neg x$, $y$,$z$]) have a branching factor of at most $\tau (4,9)$.
\medskip\noindent\newline\textbf{Case 3b :} $|Vars(\alpha\,\cup\,\beta\,\cup\,\gamma\,\cup\,\delta\,\cup\,\omega\,\cup\,\sigma)|$ = 6 

There are clauses $(x\lor  a\lor b),\, (x\lor  c),\,(y \lor  a),\,  (y \lor  d)\, (z \lor  e),\,  (z \lor  f)$. In the branch $F$[$z$], we are sure to eliminate another variable due to step 6c of the algorithm. Returning 3-occur-SAT($F$[$z$]) $\lor$ 3-occur-SAT($F$[$x$, $y$,$\neg$$z$]) have a branching factor of at most $\tau (4,9)$.
\medskip\noindent\newline\textbf{Case 3c :} $|Vars(\alpha\,\cup\,\beta\,\cup\,\gamma\,\cup\,\delta\,\cup\,\omega\,\cup\,\sigma)|$ = 5 

 Either one of $|Vars(A\,\cap\,B)| = 2$ or $|Vars(A\,\cap\,B)| = |Vars(A\,\cap\,C)| = 1$. 

If $|Vars(A\,\cap\,B)| = 2$, then there are clauses $(x\lor  a\lor b),\, (x\lor  c),\,(y \lor  a),\,  (y \lor  b)\, (z \lor  d),\,  (z \lor  e),\, (\neg a \lor \gamma)$ and  $(\neg b \lor  \delta)$. In the branch $F$[$x$], after resolution, we will have clauses $(y \lor  \gamma),\,  (y \lor  \delta)\, (z \lor  d),\,  (z \lor  e)$. If either $\gamma$ or $\delta$ contains the variable $d$ or $e$, an additional variable can be eliminated either by step 6a or 6c of the algorithm. 
Returning 3-occur-SAT($F$[$x$]) $\lor$ 3-occur-SAT($F$[$\neg$$x$, $y$, $z$]) have a branching factor of at most $\tau (5, 8)$. 
Otherwise, in the branch $F$[$z$], we can eliminate a variable 
by Lemma \ref{C = 3,reduction}. In the branch $F$[$x$, $y$, $\neg$$z$], $\neg a$ and $\neg b$ are pure literals and an additional variable can be eliminated. Hence, returning 3-occur-SAT($F$[$z$]) $\lor$ 3-occur-SAT($F$[$x$, $y$,$\neg$$z$]) have a branching factor of at most $\tau (4,9)$. 

If $|Vars(A\,\cap\,B)| = |Vars(A\,\cap\,C)| = 1$, then there are clauses $(x\lor  a\lor b),\, (x\lor  c),\,(y \lor  a),\,  (y \lor  d)\, (z \lor  b),\,  (z \lor  e),\, (\neg a \lor \gamma)$ and  $(\neg b \lor  \delta)$. Again, if  either $\gamma$ or $\delta$ contains the variable $d$ or $e$, an additional variable can be eliminated in the branch $F$[$x$] either by step 6a or 6c of the algorithm. Returning 3-occur-SAT($F$[$x$]) $\lor$ 3-occur-SAT($F$[$\neg$$x$, $y$, $z$]) have a branching factor of at most $\tau (5, 8)$. If $\gamma$ (or $\delta$) is $\neg c$, resolution via $a$ gives $(x\lor  \neg c\lor b),\, (x\lor  c),\,(y \lor  \neg c),\,  (y \lor  d)\, (z \lor  b),\,  (z \lor  e)$ and  $(\neg b \lor  \delta)$. Due to $(x\lor  c)$,  $(x\lor  \neg c\lor b)$ can be reduced to $(x\lor b)$ allowing the formula to remain 3-regular. Else, either  $\gamma$ and $\delta$ contain a variable other than $\{x,y,a,b,c,d,e\}$.In the branch $F$[$x$, $y$, $\neg$$z$], $\neg a$ and $\neg b$ are pure literals and an at least 9 variables are eliminated ($x$, $y$, $z$, $a$,  $b$, $c$, $d$, $e$ and the variables in $\gamma$ or $\delta$). Hence, returning 3-occur-SAT($F$[$z$]) $\lor$ 3-occur-SAT($F$[$x$, $y$,$\neg$$z$]) have a branching factor of at most $\tau (4,9)$.\end{proof}
\begin{lemma}\label{C =3, +} 
    If $|\alpha|+|\beta| = 2$ and one of the literals $x$, $y$ or $z$ is neighbours with a negative literal, there exists $w\in \{x,y,z\}$ such that branching factor of 3-occur-SAT($F$[$\neg w$]) $\lor$ 3-occur-SAT($F$[$w$]) is at most $\tau(4,9)$. \end{lemma}
        
\begin{proof}
        By lemma  \ref{many rules},  $\alpha, \beta, \gamma, \delta, \omega, \sigma$ all contain a different variable. Thus, in the branch $F[\neg x,y ,z]$, 9 variables are eliminated. Without loss of generality, suppose $x$ shares a clause with a negative literal and that there are clauses $(\neg x \lor \neg y \lor \neg z), (x \lor \neg v_1), (x \lor l), (v_1 \lor C_2), (v_1 \lor C_2)$. In the branch $F[x]$, $v_1$ is set to 1 and the variables in the subclauses $C_1$ and $C_2$ can be eliminated. As $C_1, C_2$ cannot contain solely the variable of $l$ (due to step 5), $C_1, C_2$ contain a variable outside of $x, v_1$ and the variable of $l$, and at least 4 variables are eliminated in this branch.
    
\end{proof}
\subsection{Tackling all-negative 2-clauses}
Lastly, we analyse the case when $|C|$ = 2. There are clauses ($\neg x\lor \neg y)$, $(x\lor \alpha)$, $(x\lor \beta)$, $(y\lor \gamma)$ and $(y\lor \delta)$. Due to Lemma \ref{many rules}, $Vars(\alpha) \cap Vars(\beta) = \emptyset$ and if any of $\alpha,\beta,\gamma,\delta$ is a singleton subclause containing variable $u$, $u$ is not present in the other 3 subclauses. 

\begin{lemma}\label{tau}
    Returning 3-occur-SAT($F$[$\neg$$x$, $y$]) $\lor$ 3-occur-SAT($F$[$x$, $\neg y$]) has a branching factor of at most $\tau(6-|N(x)|+|N(y)|, 6-|N(y)|+|N(x)|)$.\end{lemma}
\begin{proof}
The literal $x$ is in at least $4-|N(x)|$ clauses with a singleton subclause. By Lemma \ref{not x}, in the branch $F$[$\neg x$, $y$], variables $x,y$, variables in $\gamma,\delta$ and variables from $\alpha,\beta$ (if they are a singleton subclause) are eliminated. This brings a total of $2 + |N(y)| + 4-|N(x)|$ variables or $6-|N(x)|+|N(y)|$ variables. Symmetric analysis for the branch $F$[$x$, $\neg  y$].
 \end{proof}
\begin{lemma}\label{3 clause}Let  $C = (\neg x \lor \neg y)$. If there is a literal $u \in \{x,y\}$ such that $u$ in a clause of size 4 or more, returning 3-occur-SAT($F$[$\neg$$x$, $y$]) $\lor$ 3-occur-SAT($F$[$x$, $\neg y$]) has branching factor of at most $\tau(4,9)$. Furthermore, if $u$ is in a clause of size 2, then the branching factor is at most $\tau(5,8)$.\end{lemma}
\begin{proof}
    Without loss of generality, let $x$ be in a clause with $4$ or more variables.  From the analysis as Lemma \ref{tau}, the branch $F[x ,\neg  y]$ eliminates at least $6-|N(y)|+|N(x)|$ variables. Furthermore, the literal $x$ is in at least $5 -|N(x)| $ clauses with a singleton subclause. The branch $F[\neg x$, $ y]$ eliminates at least $7 -|N(x)|+|N(y)|$ variables. Thus, the branching factor is at most  $\tau(7 -|N(x)|+|N(y)|, 6-|N(y)|+|N(x)|)$. We can see that at least 13 variables are eliminated across both branches  and each branch eliminates at least 4 variables. This gives a branching factor of at most $\tau(4,9)$. 

Furthermore, if the literal $x$ is also in a clause with of size 2,  then we note that in both branches, at least 5 variables are eliminated. As, at least 13 variables are eliminated across both branches, this gives a branching factor of at most $\tau(5,8)$. 
\end{proof}
Note that by Lemma \ref{tau}, based on counting the numbers of neighbours of $x,y$, the maximum branching factor for returning 3-occur-SAT($F$[$\neg$$x$, $y$]) $\lor$ 3-occur-SAT($F$[$x$, $\neg y$]) is at most $\tau(4,8)$. Thus, if we are able to eliminate an additional variable in either branch, this brings the branching factor down to at most $\tau(4,9)$, making the branch admissible for step 7 of our algorithm.
 \begin{lemma}\label{C =2, +}  With the clause ($\neg x \lor \neg y$), if either the literals $x$ or $y$ are neighbours with a negative literal, we can eliminate an additional variable from one of the branches. \end{lemma}\begin{proof} 
Let $Elim$ = $N(x) \,\,\cup \{x,y\}\,\,\cup\,\, \{v \in N(y) | \text{there is a clause } (v\lor y) \text{ or } (\neg v \lor y)          \}$. In the branch $F[x]$, by Lemma \ref{not x}, the variables in $Elim$ are eliminated. We will argue below that if $x$ shares a clause with a negative literal ($\neg u)$, another variable outside of $Elim$ can be eliminated in the branch $F[x]$. 
Let $S \subset Elim$, be the set of variables such that all clauses containing a variable in $S$ are eliminated in the branch $F[x]$. $S$ can be iteratively generated from the following procedure. Initially, $S = \{u\}$ as the positive literal of $u$ becomes a pure literal and all clauses containing this variable are eliminated. Suppose there is a variable $v\in N(x)/S$ such that both positive instances or the negative instance of $v$ is in a clause with either the literal $x$ or the literals of a variable in $S$. In the branch $F[x]$, $v$ will become a pure literal and all clauses containing $v$ will be eliminated.  Add $v$ to $S$. Suppose there is a $v\in N(y)/S$ such that there is the clause ($y \lor v$) and all clauses containing the literal $\neg v$ contains a variable in $S$.  In the branch $F[x]$, As $v$ is set to 1 in the branch, all clauses containing the literal $v$ are eliminated. All clauses containing $\neg v$ also contains a variable in $S$, and hence all clauses containing $\neg v$ are eliminated. Hence, we can add $v$ to $S$. Repeat this process until convergence.

If any variable in $S$ is neighbours with a variable outside of $Elim$, as all clauses containing a variable in $S$ are eliminated in the branch $F[x]$, another variable can be eliminated using reductions up to step 6. Hence, we assume that all variables in $S$ are only neighbours with variables in $Elim$. As we delve into several cases, we discover that under this assumption some subset of $Elim$ are neighbours with only 3 subclauses containing variables outside the subset and hence can be reduced by step 6d. Suppose there are clauses ($y\lor v$), ($v\lor \sigma$). In the branch $F$[$x,\neg y$], $v$ is assigned to 1 and all variables in $\sigma$ are eliminated. If $\sigma$ contains a variable beyond $Elim$, an additional variable can be eliminated. Hence, we assume that $N(V) \subset Elim$. Thus for the variable $v$ only the literal $\neg v$ can be a neighbour with a subclause containing variables outside of $Elim.$ 

 The notation $v^*$ refers to either a positive or negative literal of the variable $v$ so as to distinguish literals we know are positive or negative and just a generic literal of a variable. Unless explicitly stated otherwise, when mentioning that a subset can be reduced, we are referring to a reduction by step 6d.   Let $|\alpha|\geq|\beta|$, $|\gamma|\geq|\delta|$ and $|\alpha|+|\beta|\geq|\gamma|+|\delta|$.
\medskip\noindent\newline\textbf{Case 1:} $|\alpha|=|\beta|=|\gamma|=|\delta|=2$

Without loss of generality, let there be clauses ($x \lor \neg a \lor b^*$), ($x \lor c^* \lor d^*$). $Elim = \{x,y,a,b,c,d\}$ and $a \in S$. The literal $a$ must be neighbours with both the variables $c,d$. Hence $c,d \in S$. If either the last instance of $c,d$ is neighbors with the variable $b$, $b\in S$, the set $\{x,y,a,b,c,d\}$ can be reduced. If the last instance of $c,d$ are solely neighbours with $y$, the set $\{x,y,a,c,d\}$ can be reduced.
\medskip\noindent\newline\textbf{Case 2:} $|\alpha|=|\beta|=|\gamma|=2,  |\delta|=1$

Suppose the literal $y$ is neighbours with a negative literal. Let there be clauses  $(y\lor e^*)$, $(y \lor f^* \lor g^*)$. In the branch $F[\neg x, y]$, $Elim = \{x,y,e,f,g\}$. Suppose one of $e^*, f^*$ or $g^*$ is $\neg h$. $h\in S$ and must be neighbours only with the other 2 of $e,f,g$. Thus, the set $\{x,y,e,f,g\}$ can be reduced.

Otherwise, without loss of generality, let there be clauses ($x \lor \neg a \lor b^*$), ($x \lor c^* \lor d^*$),  $(y \lor e)$. In the branch $F[x,\neg  y]$, $Elim = \{x,y,a,b,c,d,e\}$ and  $a\in S$. As the literal $a$ cannot share a clause with $y$ (otherwise  step 6a is applicable), the literal $a$ must share a clause with at least one of $c$ or $d$. Without loss of generality, let the literal $a$ be neighbours with $c^*$. Thus, $c \in S$.

If there is the clause ($y \lor c^* \lor f$), $f$ is $b^*$ as $c \in S$ ($f$ cannot be $e^*$ or $d^*$ otherwise step 5 of the algorithm would apply). Depending on if $a$ is a neighbour of $d^*$ or $e^*$, the set $\{x,y,a,b,c,d\}$ or $\{x,y,a,b,c,e\}$ can be reduced. Hence, the last instance of $c$ is not a neighbour of $y$ and must be a neighbour of $b^*$ or $e^*$. If the literal $a$ is neighbours with $d^*$, $d \in S$. One of the sets $\{x,y,a,c,d,e\}$ or $\{x,y,a,b,c,d\}$ or  $\{x,y,a,b,c,d,e\}$ (depending on if $b$ and/or $e$ is in $S$) can be reduced. Else, the literal $a$ is neighbours with $c^*$ and $e^*$ solely. If the last instance of $c$ is neighbours with $b$, $b \in S$. Either the set $\{x,y,a,b,c,e\}$ or $\{x,y,a,b,c,d,e\}$ (if $d$ in $S$) can be reduced. Else, the literal $a$ is neighbours with $c$ and $e$ solely and the last instance of $c$ is neighbours with $e$ and the set $\{x,y,a,c,e\}$ can be reduced.
\medskip\noindent\newline\textbf{Case 3:} $|\alpha|=|\beta|=2,  |\gamma|=|\delta|=1$

Let there be clauses $(y \lor e^*)$, $(y\lor f^*)$.  In the branch $F$[$x, \neg y$], $Elim = \{x,y,e,f\}$. If $e^*$ (symmetric for $f^*$) is negative, the 2 instances of the literal $e$ must be neighbours with a variable outside of $Elim$ and another variable can be eliminated.  Note that the  literal $e$ cannot be a neighbour of the literal $x$ due to Lemma \ref{many rules}. Otherwise, without loss of generality, let there be clauses ($x \lor \neg a \lor b^*$), ($x \lor c^* \lor d^*$),  $(y \lor e)$, $(y \lor f)$. $Elim = \{x,y,a,b,c,d,e,f\}$ and  $a\in S$.
\medskip\noindent\newline\textbf{Case 3a:} $a$ is neighbours with only variables $e,f$

If there are clauses ($a \lor e$), ($a \lor f$) together with ($x \lor \neg a \lor b^*$), ($x \lor c^* \lor d^*$),  $(y \lor e)$, $(y \lor f)$. Resolving via $a, x, y$ gives us clauses ($e\lor b^*$), ($f\lor b^*$), ($e\lor c^*\lor d^*$), ($f\lor c^* \lor d^*$). Introduce 2 new variables $i,j$ to have clauses ($e\lor i$), ($f\lor i$), ($e\lor j$), ($f\lor j$), ($\neg i\lor b^*$), ($\neg j\lor c^* \lor d^*$), thus decreasing the number of variables by 1.

Else, $a$ is neighbours with $\neg e$ or $\neg f$. Without loss of generality, let  $a,\neg e$ share a clause and hence $e \in S$. The last instance of $e$ is either neighbours to $f^*$ or one of $b^*,c^*,d^*$. In the former case, step 6d of the algorithm is applicable to the set $\{x,y,a,e,f\}$. In the latter case, one of $b^*,c^*,d^*$ is in $S$ (call this variable $g$). $g$ must be neighbours with either last instance of $f$ or another variable in $b,c,d$ and that variable must be in $S$. The set $\{x,y,a,e,f, g\}$ (in the former case) or $\{x,y,a,b,c,d, e,f\}$ (in the latter case as at most one variable out of $b,c,d$ is not in $S$) can be reduced.
\medskip\noindent\newline\textbf{Case 3b:} $a$ is neighbours with $c^*,d^*$

$c,d\in S$. If the last instance of $c,d$ are neighbours with $b^*$ then $b\in S$ and the set $\{x,y,a,b,c,d,e,f\}$ can be reduced.  If the last instance of $c,d$ are not neighbours with $b^*$, they must be neighbours with $e^*, f^*$, and the set $\{x,y,a,c,d,e,f\}$ can be reduced.
\medskip\noindent\newline\textbf{Case 3c:} $a$ is neighbours with $c^*, e$ and not $d^*$

$c \in S$. If the last instance of $c$ is neighbours with $b^*$, then $b \in S$, and the last instance of $b$ must be neighbours with $d^*$ or only $e^*/f^*$. Either the set $\{x,y,a,b,c,d,e,f\}$ (in the former case as $a,b,c,d$ are all in S) or $\{x,y,a,b,c,e,f\}$ (in the latter case) can be reduced.

If the last instance of $c$ is neighbours with $e$, the set $\{x,y,a,c,e\}$ or $\{x,y,a,c,e, f\}$ (if $f$ shares a clause with $a,c$ or $e$) can be reduced.

If the previous case do not apply, the last instance of $c$ is only neighbours with $f^*$. If the last instance of $f^*$ is neighbours with $a,c,e$, the set $\{x,y,a,c,e, f\}$ can be reduced. If the neighbour to $a$ is $\neg e$ then $e \in S$. 
One of $b,d$ must be a neighbour of the literal $e$ and must be in $S$. The set $\{x,y,a,b,c,d,e, f\}$ can be reduced.

Else, there are clauses ($x \lor \neg a \lor b^*$), ($x \lor c_1 \lor d^*$),  $(y \lor e)$, $(y \lor f)$, $(a \lor e)$,$(a \lor c_2)$,$(c_3 \lor f)$, where $c_1,c_2,c_3$ are instances of the variable $c$. Resolving via $x,y$ we have ($e \lor \neg a \lor b^*$), ($f \lor \neg a \lor b^*$), ($e \lor c_1 \lor d^*$),  ($f \lor c_1 \lor d^*$), $(a \lor e)$,$(a \lor c_2)$,$(c_3 \lor f)$. Resolving via $a$, we have ($e \lor b^*$),($e \lor c_2 \lor b^*$), ($f \lor e\lor b^*$), ($f \lor c_2 \lor b^*$), ($e \lor c_1 \lor d^*$),  ($f \lor c_1 \lor d^*$),  $(c_3 \lor f)$ which can be simplified to ($e \lor b^*$), ($f \lor c_2 \lor b^*$), ($e \lor c_1 \lor d^*$),  ($f \lor c_1 \lor d^*$),  $(c_3 \lor f)$. From the following cases, observe that step 6d of the algorithm is applicable to the set $\{x,y,a,c\}$.

If $c_1$ is $\neg c$, we have ($e \lor b^*$), ($f \lor c \lor b^*$), ($e \lor \neg c \lor d^*$),  ($f \lor \neg c \lor d^*$),  $(c \lor f)$. Resolving via $c$, we have ($e \lor b^*$), ($e \lor f \lor b^* \lor d^*$),  ($f \lor b^* \lor d^*$), ($e \lor f \lor d^*$),  ($f \lor  d^*$), which can be simplified to ($e \lor b^*$), ($f \lor  d^*$)

If $c_2$ is $\neg c$,  we have ($e \lor b^*$), ($f \lor \neg c \lor b^*$), ($e \lor c \lor d^*$),  ($f \lor c \lor d^*$),  $(c \lor f)$. Resolving via $c$, we have ($e \lor b^*$), ($f \lor e \lor d^* \lor b^*$), ($f \lor d^* \lor b^*$),($f \lor b^*$), which can be simplified to ($e \lor b^*$), ($f \lor b^*$).

If $c_3$ is $\neg c$, we have ($e \lor b^*$), ($f \lor c \lor b^*$), ($e \lor c \lor d^*$),  ($f \lor c \lor d^*$),  $(\neg c \lor f)$. Resolving via $c$, we have ($e \lor b^*$), ($f \lor b^*$), ($e \lor f \lor d^*$),  ($f \lor d^*$) which can be simplified to ($e \lor b^*$), ($f \lor b^*$), ($f \lor d^*$).
\medskip\noindent\newline\textbf{Case 4:} $|\alpha|=|\gamma|= 2,  |\beta|=|\delta|=1$

There are clauses ($x\lor a^*\lor b^*$), ($x\lor c^*)$, $(y\lor d^*)$, ($y \lor e^* \lor f^*$). One of $a^*,b^*$ may be $e^*,f^*$ . In the branch $F$[$x,\neg y$]. $Elim = \{x,y,a,b,c,d\}$. If one of $a^*,b^*,c^*$ is negative, the branch $F$[$x,\neg y$] can eliminate an additional variable as follows (symmetric for $d^*, e^*, f^*$ being negative). 

If $c^*$ is negative, $c\in S$. The literal $c$ must be neighbours with at least one of $a^*$ or $b^*$. Without loss of generality, let $a^*$ be neighbours with $c$ and $a \in S$. Hence, the last instance of $a$ must be neighbours with $d^*$. Depending on if $c$ is a neighbour to $b^*$ or $d^*$, the set $\{x,y,a,b,c,d\}$ or $\{x,y,a,c,d\}$ can be reduced.

If $a^*$ is negative (symmetric analysis if $b^*$ is negative), $a\in S$. Hence, $a$ must be neighbours with solely $c^*$ and $d^*$. Hence, $c\in S$ and the last instance of $c$ is neighbours with either $b^*$ or $d^*$. Depending on if the last instance of $c$ is a neighbour to $b^*$ or $d^*$, the set $\{x,y,a,b,c,d\}$ or $\{x,y,a,c,d\}$ can be reduced. 
\medskip\noindent\newline\textbf{Case 5:} $|\alpha|= 2, |\beta|=|\gamma|=|\delta|=1$

There are clauses ($x\lor a^*\lor b^*$), ($x\lor c^*)$, $(y\lor d^*)$, ($y \lor e^*$). If either $d^*$ or $e^*$ is negative, another variable can be eliminated in the branch $F$[$\neg x, y$] as follows. Without loss of generality, let $d^*$ be negative. The literal $d$ must be neighbours solely with $c^*$ and $e^*$. The set $\{x,y,c,d,e\}$ can be reduced. 

Else, we know $d^*, e^*$ are positive. If any of $a^*,b^*$ or $c^*$ is negative, another variable can be eliminated in the branch $F$[$x, \neg y$] as follows. $Elim =\{x,y,a,b,c,d,e\}$.

If $c^*$ is negative, $c \in S$. If the literal $c$ is neighbours with solely $d^*$ and $e^*$, the set $\{x,y,c,d,e\}$ can be reduced. Else, $c$ is neighbours with either $a^*$ or $b^*$. Without loss of generality, let $c$ be neighbours with the $a^*$ and $a \in S.$ As the literals $d$ and $e$ can only have neighbours in $Elim$ (otherwise another variable is eliminated as both $d$ and $e$ are assigned to 1 in the branch $F$[$x, \neg y$]), either the set $\{x,y,a,c,d,e\}$ (if $b^*$ is not neighbours with $c^*,d^*$ or $e$) or $\{x,y, a,b,c,d,e\}$ (if $b^*$ is neighbours with one of $c^*,d^*$ or $e$) can be reduced.
 
 Otherwise, without loss of generality, let $a^*$ be negative, $a\in S$. If the literal $a$ is neighbours with $c^*$, $c^* \in S$. Either the set $\{x,y, a,c,d,e\}$ (if $b^*$ is not neighbours with $c^*,d^*$ or $e$) or $\{x,y, a,b,c,d,e\}$ (if $b^*$ is neighbours with one of $c^*,d^*$ or $e$) can be reduced. 

Else, the literal $a$ is only neighbours with $d^*$ and $e^*$. If the literal $a$ is neighbours with solely the literals $d$ and $e$, there are clauses ($a \lor d$), ($a \lor e$). Resolving via $a, x, y$ gives us clauses ($d\lor b^*$), ($e\lor b^*$), ($d\lor c^*$), ($e\lor c^*$). Introduce 2 new variables $i,j$ to have clauses ($d\lor i$), ($e\lor i$), ($d\lor j$), ($e\lor j$), ($\neg i\lor b^*$), ($\neg j\lor c^*)$, thus decreasing the number of variables by 1. Otherwise, the literal $a$ is neighbours with either $\neg d$ or $\neg e$. Without loss of generality, let $a$ be neighbours with $\neg d$. The last instance of the literal $d$ must hence be neighbours with variable $b,c$ or $e$. If the last instance of $d$ is solely neighbours with the variable $e$, the set $\{x,y, a,d,e\}$ can be reduced. Else, if $d$ is neighbour with $b$ or $c$ and, the set $\{x,y,a,b,c,d,e\}$ can be reduced.
\medskip\noindent\newline\textbf{Case 6:} $|\alpha|=|\beta|=|\gamma|=|\delta|=1$

There are clauses ($x\lor a^*$), ($x\lor b^*)$, $(y\lor c^*)$, ($y \lor d^*$). If $a^*$ is negative, the other 2 instance of the variable $a$ are neighbours to only variables in $\{x,y,a,b,c,d\}$. If $a^*$ is positive, the other instance of the positive literal of $a$ must be neighbours to only variables in $\{x,y,a,b,c,d\}$ (otherwise another variable is eliminated as $a$ is assigned to 1 in the branch $F$[$\neg x, y$]). The analysis is symmetric for $b^*,c^*,d^*$. Hence, if one of $a^*,b^*,c^*$ or $d^*$ is negative the set $\{x,y,a,b,c,d\}$ can be reduced. 
\end{proof}

\section{Step 8: Making $F$ monotone}
In this section, we show that enough structure is imposed on our formula for us to use a reduction rule to make the formula monotone.  
\begin{lemma}\label{step 8}
    If $F$ is a step-7 reduced formula, then if $C$ is an all-negative clause
    (i) $|C| < 4$  
    (ii) If $|C| =3$, let $C = (\neg x \lor \neg y \lor \neg z)$. The literals $x,y,z$ are in all-positive clauses of length 2.
    (iii) If $|C| =2$, let $C = (\neg x \lor \neg y)$. The literals $x,y$ are in all-positive clauses of length 2 or 3.
\end{lemma}
\begin{proof}
    By Lemma \ref{C=4}, $|C| \leq 3$. If $C = (\neg x \lor \neg y \lor \neg z)$, by lemma \ref{C=3} and \ref{C=3ii}, we are left with the case $|N(x)| = |N(y)| = |N(z)| = 2$ and thus the literals $x,y,z$ are in clauses of length 2. Furthermore, by lemma \ref{C =3, +}, the literals $x,y,z$ are in all-positive clauses. If $C = (\neg x \lor \neg y)$, by lemma \ref{3 clause} and \ref{C =2, +}, the literals $x,y$ are in all-positive clauses of length 2 or 3. \end{proof}
\begin{definition} \emph{\textbf{(Autarkic sets)} }(Slightly modified definition from \cite{chu2021improved})
    A set of variables \emph{X} is called an autarkic set if each clause containing a negative literal in \emph{X} also contains a positive literal in \emph{X}. A CNF formula   $F$ is satisfiable if and only if it is satisfiable when all variables in an autarkic set are assigned 1.
\end{definition}

\begin{lemma}\label{aut} 
    Let $S$ be the set of all variables not in an all-negative clauses. $S$ is an autarkic set and hence $S$ is satisfiable if and only it is satisfiable when all variables in $S$ are assigned 1. 
\end{lemma}
\begin{proof}
 As variables in $S$ are not in an all negative clause, all clauses that contain a negative literal in $S$ must also contain a positive literal in $S$.  Furthermore, as positive literals in $\text{Vars}(F)/S$ are in all positive clauses by Lemma \ref{step 8}, all clauses that contain a negative literal in $S$ must also contain a positive literal in $S$. 
\end{proof}
\begin{corollary}
    If $F$ is a step-8 reduced formula, then $F$ is monotone.
\end{corollary}

\section{Step 9: Dealing with the remaining all-negative 2-clause}

For us to use the fast 3-SAT algorithm by \cite{492575}, for all-negative 2-clause $(\neg x \lor \neg y)$, it is necessary that $|N(x)| + |N(y)| \leq 6$. Without loss of generality, let $|N(x)| \geq |N(y)|$. If $F$ is a step-8 reduced formula, $|N(x)| \leq4$ by Lemma \ref{step 8}. Thus, if $|N(x)| + |N(y)| > 6$, either $|N(x)| = 4, |N(y)| =3$ or $|N(x)| = |N(y)| = 4$.  Now that the formulas are monotone, there is sufficient structure for us to tackle these cases while having a good branching factor. In Lemma \ref{monotone + no N = 3}, we tackle the former case with just our primary branching rule and in Lemma \ref{step 9}, we tackle the latter case with our secondary branching rule.

\begin{lemma} \label{C= 2, special}  With clauses ($\neg x\lor\neg y$), $(y\lor a\lor b)$, ($y\lor c$), ($\neg c \lor \neg d)$, either a reduction apply and we can decrease the number of variables in $F$, or one of 3-occur-SAT($F$[$\neg$$x$, $y$]) $\lor$ 3-occur-SAT($F$[$x$, $\neg y$]) or 3-occur-SAT($F$[$\neg$$c$, $d$]) $\lor$ 3-occur-SAT($F$[$c$, $\neg d$]) has branching factor of at most $\tau(5,8)$. \end{lemma}
\begin{proof}
    Let there be the clause $(c \lor \gamma)$. We shall first analyse the branching factor of 3-occur-SAT($F$[$\neg$$c$, $d$]) $\lor$ 3-occur-SAT($F$[$c$, $\neg d$]). If $|\gamma| \geq 3$, the variables in $\gamma$ are eliminated in the branch $F[c]$. If $\gamma$ contains a negative literal, by Lemma \ref{C =2, +}, an additional variable is eliminated  in the branch $F[c]$. In both cases, the variable $y$ is eliminated in both branches.
Expanding on the idea of Lemma \ref{tau}, at least 13 variables are eliminated across both branches and each branch eliminates at least 5 variables. Hence, the branching factor is at most $\tau(5,8)$. Otherwise, assume $|\gamma| \leq 2$ and $\gamma$ does not contain a negative literal.

We will now show that an additional variable is eliminated in the branch $F$[$x$, $\neg y$], leading to a branching factor of at most $\tau(5,8)$. Note that $d \notin \gamma$, otherwise step 5 of the algorithm would be applicable. If $d$ is neighbours with $y$, step 6c of the algorithm applies due to the clause ($\neg c \lor \neg d)$. If $\gamma$ contains a variable not in $N(x)$ or if $d$ is not a variable in $N(x)$, another variable can be eliminated in this branch. Otherwise, we obtain the clauses ($\neg x\lor\neg y$), $(x\lor d \lor \sigma)$, $(x\lor\omega)$, $(y\lor \alpha)$, ($y\lor c$), ($\neg c \lor \neg d$), ($c\lor\gamma$). If $x$ is neighbours with a negative literal, by Lemma \ref{C =2, +},  returning 3-occur-SAT($F$[$\neg$$x$, $y$]) $\lor$ 3-occur-SAT($F$[$x$, $\neg y$]) has branching factor of at most $\tau(5,8)$. 

Suppose $|\gamma|=1$ and let $\gamma$ be $e$, $e \in N(x)$ and $e\notin N(y)$. Resolution of the clauses via $x,y$ gives us $(d \lor \sigma\lor \alpha)$, ($d \lor \sigma\lor c$), $(\omega\lor \alpha)$, ($\omega\lor c$), ($\neg c \lor \neg d$),  ($c\lor e$). Resolution via $c$, gives us
$(d \lor \sigma\lor \alpha)$, ($d \lor \sigma\lor \neg d$), $(\omega\lor \alpha)$, ($\omega\lor \neg d$), ($\neg d\lor e$) which can be simplified to $(d \lor \sigma\lor \alpha)$, $(\omega\lor \alpha)$, ($\omega\lor \neg d$), ($\neg d\lor e$). With the clause ($d \lor \beta$), resolution of $d$ yields the clauses, $(\omega\lor \alpha)$, ($\omega\lor \sigma\lor \alpha$), ($\sigma\lor \alpha\lor e$), ($\omega\lor \beta$), ($\beta\lor e$) which can be simplified to $(\omega\lor \alpha)$, ($\sigma\lor \alpha\lor e$), ($\omega\lor \beta$), ($\beta\lor e$). We need to introduce 4 variables to deal with the extra occurrences of $e, \beta,\omega$ and $\alpha$ for the safe resolution of the variables $\{x,y,c,d\}$. Hence, if we can eliminate one of the above clauses or eliminate an instance of $e$, step 6d is applicable to $\{x,y,c,d\}$. If the variable $e$ is in $\sigma$, ($\sigma\lor \alpha\lor e$) can be simplified to ($\sigma\lor \alpha$). If the literal $e$ is in $\omega$, ($\omega\lor \beta$) can be eliminated as it subsumes ($\beta\lor e$).

Suppose $|\gamma|=2$ and let $\gamma$ contain the variables $e,f$, and  $\{e,f\} \subseteq N(x)$ and $e,f \notin N(y)$.  In the branch $F[x,\neg y]$, the literals $\neg e$ and $\neg f$ are pure literals and if they share  a clause with a variable outside $N(x)$, another variable can be eliminated in this branch. If literals $\neg e$ or $\neg f$ is neighbours with $d$, step 6a of the algorithm applies. As $x$ is already neighbours with $d,e,f$, by Lemma \ref{3 clause}, $x$ can only be neighbours with at most 1 other variable. Let that variable be $g$. If $\neg e, \neg f$ are both solely neighbours with the variable $g$, step 6d of the algorithm applies to $\{x,y,c,e,f,g\}$ due to the criteria of Lemma \ref{key reduction}.  
\end{proof}
\begin{lemma}\label{monotone + no N = 3} 
       If $F$ is a step-8 reduced formula with the clause $(\neg x \lor \neg y)$, then $|N(x)| \neq 3$ and $|N(y)| \neq 3$.
\end{lemma}

\begin{proof}
    Without loss of generality, assume $|N(x)| = 3$ and there are clauses ($\neg x\lor\neg y$), ($x\lor a$). As $F$ is monotone, the literal $\neg a$ is in a all-negative clause. If the literal $\neg a$ is in an all-negative 2-clause, by Lemma \ref{C= 2, special}, such a clause will have a good branching factor for step 7 of the algorithm. Otherwise, by Lemma \ref{step 8}, literal $\neg a$ is in an all-negative 3-clause.
    
    We have the clauses $(\neg a\lor\neg u\lor\neg v)$, $(a\lor x)$, $(a\lor b)$, $(u\lor c)$, $(u\lor d)$, $(v\lor e)$, $(v\lor f)$, $(\neg x\lor\neg y)$, ($x\lor\gamma$) with $|\gamma| = 2$.  If $y$ is $b$, resolution via $x$ yield clauses $(a\lor \neg b)$, $(a\lor b)$ and we can assign $a$ to 1. If $y$ is one of $\{c,d,e,f\}$, without loss of generality, let $y$ be $c$. In the branch $F$[$v$], if no other variables other than $v,e,f$ are eliminated, we have clauses $(\neg a\lor\neg u)$, $(a\lor x)$, $(a\lor b)$, $(u\lor c)$, $(u\lor d)$, $(\neg x \lor\neg c)$, $(x \lor \gamma), (c \lor \delta)$. After resolving via $a,u$, we have $(x\lor c)$, $(x\lor d)$, $(b\lor c)$, $(b\lor d)$, $(\neg x \lor\neg c)$, $(x \lor \gamma), (c \lor \delta)$. After resolving via $x$, we have $(\neg c\lor c)$ (redundant), $(\neg c\lor d)$, $(b\lor c)$, $(b\lor d)$,  $(\neg c \lor \gamma), (c \lor \delta)$. After resolving via $c$, we have $(\gamma \lor \delta)$, $(b\lor \gamma)$, $(d \lor \delta)$, $(b\lor d)$. As $b,d,\gamma$ and $\delta$ each occur twice, a safe resolution of $\{a,u,x,c\}$ requires the introduction of 4 new variables. If any clause can be eliminated, we can safely resolve the formula by  $\{a,u,x,c\}$ while decreasing the number of variables, leading to a branching factor of 3-occur-SAT($F$[$v$]) $\lor$ 3-occur-SAT($F$[$a$, $u$, $\neg$$v$]) to be at most $\tau(4,9)$. If $\gamma$ contains the variable $b$ or $d$ another clause can be eliminated. Note that if $\gamma$ contains $\neg d$, let $\gamma = \neg d \lor \sigma$. $(\neg d \lor \sigma \lor \delta)$,  $(d \lor \delta)$,  $(b\lor \neg d \lor \sigma)$, $(b\lor d)$ can be simplified to  $(b\lor d)$,  $(d \lor \delta)$, $(\sigma \lor \delta),$  $(b\lor \sigma)$ as the last 2 clauses are resolvents by $d$ and they are subsumed by $(\neg d \lor \sigma \lor \delta)$,  $(b\lor \neg d \lor \sigma)$. $d$ becomes a pure literal and can be eliminated. Else, $\gamma$ contains 2 variables not in $\{x,u,a,b,c,d\}$. The branch $F$[$\neg$$a$, $u$]) assigns $x$ to 1 and thus eliminates 8 variable including the 2 in $\gamma$.  Returning 3-occur-SAT($F$[$a$]) $\lor$ 3-occur-SAT($F$[$\neg$$a$, $u,v$]) has branching factor of at most $\tau(3,11)$. 

Otherwise, $y \notin \{a,b,c,d,e,f\}$. Let $Elim =\{x,y,u,v,a,b,c,d,e,f\}$. We note that for $w \in \{a,u,v\}$, assigning $w$ to 0 eliminates all variables in $Elim$ for a total of 10 variables. We note that if $a$ is set to 0, $x,b$ are assigned to 1, if $u$ is set to 0, $c,d$ are assigned to 1 and if $v$ is set to 0, $e,f$ are assigned to 1. Thus, if the literal $x,b,c,d,e,f$ are in a clause with a variable outside of $Elim$, there is a $w \in \{a,u,v\}$ such that returning 3-occur-SAT($F$[$w$]) $\lor$ 3-occur-SAT($F$[$\neg$$w$]) has branching factor of at most $\tau(3,11)$. Otherwise, $\gamma$ contains literals $m,n \in \{b,c,d,e,f\}$ (note that $\gamma$ cannot contain variables in $\{x,y,a,u,v\}$).  If $\neg m$ is in a clause with a variable outside of $Elim$ and $m \in \{c,d,e,f\}$, in the branch $F[\neg a, u,v]$, $\neg m$ becomes a pure literal and that variable can be eliminated and branching on $a$ yields a branching factor of $\tau(3,11)$. If $\neg m$ is in a clause with a variable outside of $Elim$ and $m \in \{b\}$, then suppose without loss of generality that $n$ is $c$. In the branch $F[a,\neg u,v],$ the literal $c$ is assigned to 1 and hence $\neg b$ is a pure literal. Thus, branching on $u$ yields a branching factor of $\tau(3,11)$. Otherwise, at most 1 literal in $\{\neg b, \neg c, \neg d, \neg e, \neg f\}$ shares a clauses with a variable outside of $Elim$. As literals $b,c,d,e,f$ are in clauses that solely contain variables in $Elim$, step 6.4 of the algorithm is applicable to $\{a,u,v,x,b,c,d,e,f\}$.
\end{proof}

\begin{lemma}\label{step 9}If F is a step 9 reduced formula, then F does not contain the clause ($\neg x \lor \neg y$), $|N(x)|=|N(y)| = 4$.\end{lemma} 
\begin{proof} Suppose there are clauses $(\neg x\lor\neg y)$, $(x\lor a\lor b)$, $(x\lor c\lor d)$. By Lemma \ref{step 8} and Lemma \ref{monotone + no N = 3},  the literals $\neg a$, $\neg b$, $\neg c$ and $\neg d$ are in all-negative 2-clauses and the literals $a,b,c,d$ are in all-positive 3-clauses. Note that when $|N(x)|=|N(y)|=4$, by Lemma \ref{tau}, the primary branching factor is $\tau(6,6)$.
\medskip\noindent\newline\textbf{Case 1:} There are clauses ($\neg a \lor \neg c$), ($\neg b \lor \neg d$)

 If none of $a,b,c,d$ are neighbours with $y$, then in the branch $F$[$\neg x,y$], we have ($\neg a \lor \neg c$), ($\neg b \lor \neg d$), $(a\lor b)$, $(c\lor d)$. By Lemma \ref{monotone + no N = 3}, as the last instance of $a$ cannot be in a 2-clause, we either eliminate an additional variable or there is a secondary branching with branching factor at most $\tau(5,8)$. Hence the combined branching factor is at most $\tau(6,6 + 5,6 + 8) = \tau(6,11,14)$.

 If 1 or 2 of $a,b,c,d$ are neighbours with $y$, then in the branch $F$[$\neg x,y$],  ($\neg a \lor \neg c$), ($\neg b \lor \neg d$), $(a\lor b)$, $(c\lor d)$ can be further reduced. If only $a$ (symmetric for $b,c,d$) is neighbours with $y$, the clauses simplify to ($b \lor \neg c$), ($\neg b \lor \neg d$), $(c\lor d)$ and an additional variable can be eliminated due to step 6a of the algorithm. If $a,b$ are neighbours of $y$, the clauses simplify to ($\neg c \lor \neg d$), $(c\lor d)$ and an additional variable can be eliminated due to step 5 of the algorithm. If $a,c$ are neighbours of $y$, the clauses simplify to  ($b \lor d$), ($\neg b \lor \neg d$), and an additional variable can be eliminated due to step 5 of the algorithm. Else if more than 2 of $a,b,c,d$ are neighbours with $y$, step 6d of the algorithm is applicable to the set $\{x,y,a,b,c,d\}$.
\medskip\noindent\newline\textbf{Case 2:} There is the clause ($\neg a \lor \neg c$)
 
We will have clauses ($\neg b \lor \neg e$), ($\neg d \lor \neg f$). If $e$ is not a neighbour of $a$ or $c$, in the branch $F$[$\neg b,e$], we will have clauses ($x \lor a$), ($x\lor c \lor \gamma$), ($\neg a \lor \neg c$) and we can eliminate another variable in this branch with step 6c of the algorithm. If $e$ is neighbours with one of $a$ or $c$, in the branch $F$[$\neg b,e$], after resolution of $a/c$, both clauses containing the literal $x$ contains the other of the variable $a/c$ and by step 5 of the algorithm another variable can be eliminated. 
\medskip\noindent\newline\textbf{Case 3:} There are clauses ($\neg a \lor \neg e$), ($\neg c \lor \neg f$), $e,f\in N(y)$
 
Note that because of case 1 and 2, $a,c\notin N(y)$, $e,f\notin N(x)$. If $c,f$ are both not neighbour of $a$, in the branch $F$[$a,\neg e$] and after the resolution of $x$, there will be clauses ($\neg c \lor \neg f$), ($\neg y \lor \ c \lor d$), ($f \lor y \lor \gamma$) and another variable can be eliminated by step 6a of the algorithm. Else, as only one of $c,f$ can be neighbours with $a$, after resolution  of the neighbour of $a $ in $\{c,f\}$  there will be 2 clauses with both the variables $y,c$ (if $f\in N(a)$) or $y,f$ (if $c\in N(a)$) and another variable can be eliminated by step 5 of the algorithm.

Note that if case 3 is not applicable, all the variables from one of $\{a,b\}$ or $\{c,d\}$ are in all-negative 2-clause with variables not from $N(y)$. 
\medskip\noindent\newline\textbf{Case 4:} There is the clause ($\neg a \lor \neg e$),($\neg b \lor \neg f$), $e,f\notin N(y)$
If either $a$ or $b$ are neighbours with $y$, without loss of generality let $a$ be neighbours with $y$. In the branch $F$[$\neg b,f$], there will be clauses $(\neg x \lor \neg y), (x \lor a), (y \lor a \lor \gamma)$  we can eliminate another variable by step 6c of the algorithm. 

Else, if both $a,b\notin N(y)$, in the branch $F$[$\neg x,y$], we have clauses ($\neg a \lor \neg e$),($\neg b \lor \neg f$), ($a\lor b$).  By Lemma \ref{C= 2, special}, as the last instance of $a$ cannot be in a 2-clause, we can either eliminate an additional variable or there is a secondary branching with branching factor at most $\tau(5,8)$. Hence the combined branching factor is at most $\tau(6,6 + 5,6 + 8) = \tau(6,11,14)$.
\end{proof}

\section{Final Results}
\begin{lemma}
     An instance of 3-SAT with $t$ clauses of length 3 can be solved by \cite{492575} algorithm in time $O^*(1.3645^t)$.
\end{lemma}

\begin{lemma}\label{finale}  If F is a step-9 reduced formula, then the fast 3-SAT algorithm by Beigel and Eppstein is applicable and runs in $O^*(1.1092^n)$. \end{lemma}
\begin{proof} We will show that a step 9 reduced formula has only clauses of length 2 and 3 and only $n/3$ clauses of length 3. Hence, the 3-SAT algorithm by \cite{492575} is applicable and runs in $O^*(1.3645^t) = O^*(1.3645^{(n/3)}) = O^*(1.1092^n)$.

As a step-9 reduced formula only contains monotone clauses (Lemma \ref{monotone + no N = 3}), all variables are in an all-negative clause of length 2 or 3. Hence, by Lemma \ref{step 8}, there are only clauses of length 2 and 3, thus making $F$ an instance of 3-SAT. Let $n_2$ and $n_3$ be the number of all-negative 2-clause and 3-clause respectively and let $C$ be an all-negative clause. If $|C| =2$, let $C$ be ($\neg x \lor \neg y )$. By Lemma \ref{monotone + no N = 3} and Lemma \ref{step 9}, the literals $x$ and $y$ are part of at most 2 3-clauses in total. Thus, there are at most $2n_2/3$ all-positive 3-clauses as only variables from all-negative 2-clauses can be part of all-positive 3-clauses. If $|C| =3$, let $C$ be ($\neg x \lor \neg y \lor \neg z)$. By Lemma \ref{step 8}, the literals $x,y$ and $z$ are part of no 3-clauses.  The number of variables, $n$, is equal to $2n_2+3n_3$.  As there are $n_3$ all-negative 3-clauses, there are at most $2/3 n_2  + n_3 = (2n_2+3n_3)/3 = n / 3$ clauses of length 3.
\end{proof}
\begin{theorem} If F is an instance of 3-occur-SAT, our algorithm decides the satisfiability of F in time $O^*(1.1199^n)$
\end{theorem}
\begin{lemma}\label{reduction}
    Using the reverse resolution method from \cite{wahlstrom2005faster},   an $O^*(\alpha^n)$ algorithm for 3-occur-SAT can be extended to an $O^*(\alpha^{(d-2)n})$ algorithm for SAT.
\end{lemma}
\begin{proof}
Let $F$ be a CNF formula with $n$ variables. We can construct (in polynomial time) a formula $F'$ that is equisatisfiable to $F$ such that $F'$ has $(d-2)n$ variables and a maximum degree of 3. 

Let $u$ be a variable in $F$ with $d^*$ occurrences. Without loss of generality, assume the literal $u$ has at least as many occurrences as the literal $\neg u$. If $d^*> 3$, let there be clauses ($u \lor \alpha$) and ($u \lor \beta$).  We can reduce the number of occurrences of $u$ by one by introducing a new variable $v$ to obtain the clauses ($v \lor \alpha$), ($v \lor \beta$) and ($\neg v \lor u$). After introducing $d-3$ variables, we reduced the occurrences of $u$ back to 3. If $d^*\leq2$, $u$ can be trivially dealt with either by assigning $u$ to 1 or by resolution. Let 
$rep(u)$ be the number of variables that $u$ can be represented by in $F'$.

If $d^* \geq 2$, $rep(u) = d^* - 2$ and if $d^* = 1$, $rep(u) = 0$. If $d$ is the maximum degree of $F$, the number of variables in $F'$ is  

\begin{align*}
    \sum_{u \in Vars(F)} rep(u) &\leq \sum_{u \in Vars(F)} d-2 \text{ (as } rep(u) \leq d -2 \text{)}\\
    &= (d-2)n
\end{align*}
$$$$

If $d$ is the average degree of $F$, and if each variable occurs at least twice, the number of variables in $F'$ is  $$\sum_{u \in Vars(F)} rep(u) = \sum_{u \in Vars(F)} deg(u) -2 = (d-2)n$$
\end{proof}
\begin{corollary}
        If F is an instance of SAT, our algorithm can be extended to decide the satisfiability of F in time $O^*(1.1199^{(d-2)n})$
\end{corollary}

\section{Conclusion}

We improved the best known upper bound for 3-occur-SAT to $O^*(1.1199^n)$ and for SAT to $O^*(1.1199^{(d-2)n})$. This was achieved primarily through deeper case analysis and a new powerful reduction rule for this setting. This reduction rule enables us to resolve many variables at once under a relatively broad criteria while generalising many previously-known specialised rules. We also show how to make a formula monotone without the use of branching, and suggest how to use secondary branching to avoid previously pessimal cases.

For future research, fine-grained reductions \cite{cygan2016problems,williams2018some} can be used to investigate how upper bounds for different variants of SAT relate to each other. Furthermore, as the performance of 3-occur-SAT algorithms are often the bottleneck case for algorithms measured by the formula length $L$, it could be fruitful to revisit those problems.
\bibliographystyle{plain}
\bibliography{name}

\begin{thebibliography}{10}

\bibitem{492575}
R.~Beigel and D.~Eppstein.
\newblock 3-coloring in time 0(1.3446/sup n/): a no-mis algorithm.
\newblock In {\em Proceedings of IEEE 36th Annual Foundations of Computer Science}, pages 444--452, 1995.

\bibitem{belova2020algorithms}
Tatiana Belova and Ivan Bliznets.
\newblock Algorithms for (n, 3)-maxsat and parameterization above the all-true assignment.
\newblock {\em Theoretical Computer Science}, 803:222--233, 2020.

\bibitem{brilliantov2023improved}
Kirill Brilliantov, Vasily Alferov, and Ivan Bliznets.
\newblock Improved algorithms for maximum satisfiability and its special cases.
\newblock In {\em Proceedings of the AAAI Conference on Artificial Intelligence}, volume~37, pages 3898--3905, 2023.

\bibitem{chen2009improved}
Jianer Chen and Yang Liu.
\newblock An improved sat algorithm in terms of formula length.
\newblock In {\em Workshop on Algorithms and Data Structures}, pages 144--155. Springer, 2009.

\bibitem{chu2021improved}
Huairui Chu, Mingyu Xiao, and Zhe Zhang.
\newblock An improved upper bound for sat.
\newblock {\em Theoretical Computer Science}, 887:51--62, 2021.

\bibitem{cook1971complexity}
Stephen~A. Cook.
\newblock The complexity of theorem proving procedure.
\newblock In {\em Proc. 3rd Symp. on Theory of Computing}, pages 151--158, 1971.

\bibitem{cygan2016problems}
Marek Cygan, Holger Dell, Daniel Lokshtanov, D{\'a}niel Marx, Jesper Nederlof, Yoshio Okamoto, Ramamohan Paturi, Saket Saurabh, and Magnus Wahlstr{\"o}m.
\newblock On problems as hard as {CNF-SAT}.
\newblock {\em ACM Transactions on Algorithms (TALG)}, 12(3):1--24, 2016.

\bibitem{davis1960computing}
Martin Davis and Hilary Putnam.
\newblock A computing procedure for quantification theory.
\newblock {\em Journal of the ACM (JACM)}, 7(3):201--215, 1960.

\bibitem{heule2011cube}
Marijn~JH Heule, Oliver Kullmann, Siert Wieringa, and Armin Biere.
\newblock Cube and conquer: Guiding cdcl sat solvers by lookaheads.
\newblock In {\em Haifa Verification Conference}, pages 50--65. Springer, 2011.

\bibitem{hirsch1998two}
Edward~A Hirsch.
\newblock Two new upper bounds for sat.
\newblock In {\em SODA}, pages 521--530, 1998.

\bibitem{hirsch2000new}
Edward~A Hirsch.
\newblock New worst-case upper bounds for sat.
\newblock {\em Journal of Automated Reasoning}, 24(4):397--420, 2000.

\bibitem{impagliazzo2001complexity}
Russell Impagliazzo and Ramamohan Paturi.
\newblock On the complexity of k-sat.
\newblock {\em Journal of Computer and System Sciences}, 62(2):367--375, 2001.

\bibitem{KULLMANN19991}
O.~Kullmann.
\newblock New methods for 3-sat decision and worst-case analysis.
\newblock {\em Theoretical Computer Science}, 223(1):1--72, 1999.

\bibitem{kullmann1997deciding}
Oliver Kullmann and Horst Luckhardt.
\newblock Deciding propositional tautologies: Algorithms and their complexity.
\newblock {\em preprint}, 82, 1997.

\bibitem{liu2018chain}
Sixue Liu.
\newblock Chain, generalization of covering code, and deterministic algorithm for k-sat.
\newblock In {\em 45th International Colloquium on Automata, Languages, and Programming (ICALP 2018)}. Schloss Dagstuhl-Leibniz-Zentrum fuer Informatik, 2018.

\bibitem{marques2008practical}
Joao Marques-Silva.
\newblock Practical applications of boolean satisfiability.
\newblock In {\em 2008 9th International Workshop on Discrete Event Systems}, pages 74--80. IEEE, 2008.

\bibitem{inproceedings}
Junqiang Peng and Mingyu Xiao.
\newblock Fast algorithms for sat with bounded occurrences of variables.
\newblock In {\em Proceedings of the Thirty-Second International Joint Conference on Artificial Intelligence}, pages 2004--2012, 08 2023.

\bibitem{article}
Junqiang Peng and Mingyu Xiao.
\newblock Further improvements for sat in terms of formula length.
\newblock {\em Information and Computation}, 294:105085, 08 2023.

\bibitem{scheder2022ppsz}
Dominik Scheder.
\newblock Ppsz is better than you think.
\newblock In {\em 2021 IEEE 62nd Annual Symposium on Foundations of Computer Science (FOCS)}, pages 205--216. IEEE, 2022.

\bibitem{tovey1984simplified}
Craig~A Tovey.
\newblock A simplified np-complete satisfiability problem.
\newblock {\em Discrete applied mathematics}, 8(1):85--89, 1984.

\bibitem{van1988satisfiability}
Allen Van~Gelder.
\newblock A satisfiability tester for non-clausal propositional calculus.
\newblock {\em Information and Computation}, 79(1):1--21, 1988.

\bibitem{wahlstrom2005algorithm}
Magnus Wahlstr{\"o}m.
\newblock An algorithm for the sat problem for formulae of linear length.
\newblock In {\em European Symposium on Algorithms}, pages 107--118. Springer, 2005.

\bibitem{wahlstrom2005faster}
Magnus Wahlstr{\"o}m.
\newblock Faster exact solving of sat formulae with a low number of occurrences per variable.
\newblock In {\em Theory and Applications of Satisfiability Testing: 8th International Conference, SAT 2005, St Andrews, UK, June 19-23, 2005. Proceedings 8}, pages 309--323. Springer, 2005.

\bibitem{williams2018some}
Virginia~Vassilevska Williams.
\newblock On some fine-grained questions in algorithms and complexity.
\newblock In {\em Proceedings of the international congress of mathematicians: Rio de janeiro 2018}, pages 3447--3487. World Scientific, 2018.

\bibitem{xu2019resolution}
Chao Xu, Jianer Chen, and Jianxin Wang.
\newblock Resolution and linear {CNF} formulas: improved (n, 3)-maxsat algorithms.
\newblock {\em Theoretical Computer Science}, 774:113--123, 2019.

\bibitem{cas}
Edward Zulkoski, Curtis Bright, Albert Heinle, Ilias Kotsireas, Krzysztof Czarnecki, and Vijay Ganesh.
\newblock Combining sat solvers with computer algebra systems to verify combinatorial conjectures.
\newblock {\em Journal of Automated Reasoning}, 58, 03 2017.

\end{thebibliography}
\end{document}